\documentclass{article}
\usepackage{spconf_arxiv,amsmath,graphicx}

\usepackage{bbm}

\usepackage[utf8]{inputenc} 
\usepackage[T1]{fontenc}    
\usepackage{url}            
\usepackage{booktabs}       
\usepackage{amsfonts}       
\usepackage{nicefrac}       
\usepackage{microtype}      

\usepackage{xparse}

\usepackage[ruled,vlined,linesnumbered]{algorithm2e}

\usepackage{amsmath}
\usepackage{amstext}
\usepackage{amsthm,amssymb}
\usepackage{graphicx}
\usepackage{hyphenat}
 \usepackage{color}
 \usepackage{cite}
 
\usepackage[T1]{fontenc}
\usepackage[utf8]{inputenc}
\usepackage{authblk}

\newcommand{\abs}[1]{\ensuremath{\left| #1 \right|}}

\NewDocumentCommand\col{g}{%
  \IfNoValueTF{#1}{\ensuremath{\mathrm{vec}}}{\ensuremath{\mathrm{vec}}\of{#1}}%
}

\NewDocumentCommand\of{og}{%
  \IfNoValueTF{#1}%
    { \IfNoValueTF{#2}{}{\!\({#2}\)} }%
    { \IfNoValueTF{#2}{\!\[{#1}\]}{\!\{{#2}\}} }%
}

\DeclareMathOperator{\Diff}{\ipaclap{D}{\raisebox{.204em}{\textpalhook}\kern.44em}\kern-.1em}
\NewDocumentCommand\diff{g}{%
  \IfNoValueTF{#1}
  {\text{\texthtd}}
  {\text{\texthtd}\of{#1}}%
}

\RenewDocumentCommand\ln{g}{%
  \IfNoValueTF{#1}{\mathrm{ln\ }}{\mathrm{ln}\of{#1}}%
}

\NewDocumentCommand\Real{og}{%
  \IfNoValueTF{#1}%
    { \IfNoValueTF{#2}{\mathcal{R}\!\!\mathpzc{e}}{\mathcal{R}\!\!\mathpzc{e}\!\{{#2}\}} }%
    { \IfNoValueTF{#2}{\mathcal{R}\!\!\mathpzc{e}\!\[{#1}\]}{\mathcal{R}\!\!\mathpzc{e}\!\({#2}\)} }%
}
\NewDocumentCommand\Imag{og}{%
  \IfNoValueTF{#1}%
    { \IfNoValueTF{#2}{\mathcal{I}\!\!\mathpzc{m}}{\mathcal{I}\!\!\mathpzc{m}\!\{{#2}\}} }%
    { \IfNoValueTF{#2}{\mathcal{I}\!\!\mathpzc{m}\!\[{#1}\]}{\mathcal{I}\!\!\mathpzc{m}\!\({#2}\)} }%
}

\RenewDocumentCommand\cos{g}{%
  \IfNoValueTF{#1}{\mathrm{cos}}{\mathrm{cos}\of{#1}}%
}
\RenewDocumentCommand\sin{g}{%
  \IfNoValueTF{#1}{\mathrm{sin}}{\mathrm{sin}\of{#1}}%
}
\RenewDocumentCommand\tan{g}{%
  \IfNoValueTF{#1}{\mathrm{tan}}{\mathrm{tan}\of{#1}}%
}
\RenewDocumentCommand\arccos{g}{%
  \IfNoValueTF{#1}{\mathrm{arccos}}{\mathrm{arccos}\of{#1}}%
}
\RenewDocumentCommand\arcsin{g}{%
  \IfNoValueTF{#1}{\mathrm{arcsin}}{\mathrm{arcsin}\of{#1}}%
}
\RenewDocumentCommand\arctan{g}{%
  \IfNoValueTF{#1}{\mathrm{arctan}}{\mathrm{arctan}\of{#1}}%
}
\RenewDocumentCommand\cot{g}{%
  \IfNoValueTF{#1}{\mathrm{cot}}{\mathrm{cot}\of{#1}}%
}

\newcommand{\floor}[1]{\ensuremath{\left\lfloor #1 \right\rfloor}}

\NewDocumentCommand\tr{g}{%
  \IfNoValueTF{#1}{\mathrm{tr}}{\mathrm{tr}\of{#1}}%
}

\NewDocumentCommand\diag{og}{%
  \IfNoValueTF{#1}%
    { \IfNoValueTF{#2}{\ensuremath{\mathrm{diag}}}{\ensuremath{\mathrm{diag}\of{#2}}} }%
    { \IfNoValueTF{#2}{\ensuremath{\mathrm{diag}\of[#1]}}{\ensuremath{\mathrm{diag}\of[]{#2}}} }%
}

\RenewDocumentCommand\exp{g}{%
  \IfNoValueTF{#1}{\ensuremath{\mathrm{exp}}}{\ensuremath{\mathrm{exp}}\of{#1}}%
}

\newcommand{\E}[2][]{\Operator[#1]{E}{#2}}

\NewDocumentCommand\C{g}{%
  \IfNoValueTF{#1}{\mathrm{Cov}}{\mathrm{Cov}\of{#1}}%
}

\renewcommand{\(}{\ensuremath{\left(}}
\renewcommand{\)}{\ensuremath{\right)}}
\renewcommand{\[}{\ensuremath{\left[}}
\renewcommand{\]}{\ensuremath{\right]}}

\let\oldBracketLeft\{
\let\oldBracketRight\}
\renewcommand{\{}{\ensuremath{\left\oldBracketLeft}}
\renewcommand{\}}{\ensuremath{\right\oldBracketRight}}

\NewDocumentCommand\F{og}{%
  \IfNoValueTF{#1}%
    { \IfNoValueTF{#2}{\mathcal{F}}{\mathcal{F}\!\{{#2}\}} }%
    { \IfNoValueTF{#2}{\mathcal{F}\!\[{#1}\]}{\mathcal{F}\!\({#2}\)} }%
}
\NewDocumentCommand\FInv{og}{%
  \IfNoValueTF{#1}%
    { \IfNoValueTF{#2}{\mathcal{F}^{-1}}{\mathcal{F}^{-1}\!\{{#2}\}} }%
    { \IfNoValueTF{#2}{\mathcal{F}^{-1}\!\[{#1}\]}{\mathcal{F}^{-1}\!\({#2}\)} }%
}

\NewDocumentCommand\rect{g}{%
  \IfNoValueTF{#1}
  {\ensuremath{\mathrm{rect}}}
  {\ensuremath{\mathrm{rect}\of{#1}}}%
}

\NewDocumentCommand\sinc{g}{%
  \IfNoValueTF{#1}
  {\ensuremath{\mathrm{sinc}}}
  {\ensuremath{\mathrm{sinc}\of{#1}}}%
}

\NewDocumentCommand\supp{g}{%
  \IfNoValueTF{#1}
  {\ensuremath{\mathrm{supp}}}
  {\ensuremath{\mathrm{supp}\of{#1}}}%
}

\let\oldMathcal\mathcal
\renewcommand{\mathcal}[1]{\ensuremath{\oldMathcal{#1}}}

\makeatletter
\def\foreach#1#2#3{%
  \@test@foreach{#1}{#2}#3,\@end@token
}

\def\@swallow#1{}

\def\@test@foreach#1#2{%
  \@ifnextchar\@end@token%
    {\@swallow}%
    {\@foreach{#1}{#2}}%
}

\def\@foreach#1#2#3,#4\@end@token{%
  #1{#2}{#3}%
  \@test@foreach{#1}{#2}#4\@end@token%
}
\makeatother

\newtheorem{theorem}{Theorem}[section]
\newtheorem{lemma}[theorem]{Lemma}

\newenvironment{definition}[1][Definition]{\begin{trivlist}
\item[\hskip \labelsep {\bfseries #1}]}{\end{trivlist}}

\def\l|{\left|}
\def\r|{\right|}
\def\l({\left(}
\def\r){\right)}
\def\l[{\left[}
\def\r]{\right]}

\renewcommand{\E}[1]{\mathbb{E}\left[ #1 \right]}

\def\X{{\mathbf X}}
\def\Y{{\mathbf Y}}

\begin{document} 
\title{Scalable Mutual Information Estimation using Dependence Graphs}
\name{\hspace{2cm}Morteza Noshad, Yu Zeng, Alfred O. Hero III\sthanks{This research was partially supported by ARO grant W911NF-15-1-0479.}}
\address{University of Michigan, Electrical Engineering and Computer Science, Ann Arbor, Michigan, U.S.A}
\maketitle
\begin{abstract}
The Mutual Information (MI) is an often used measure of dependency between two random variables utilized in information theory, statistics and machine learning. Recently several MI estimators have been proposed that can achieve parametric MSE convergence rate. However, most of the previously proposed estimators have high computational complexity of at least $O(N^2)$.
We propose a unified method for empirical non-parametric estimation of general MI function between random vectors in $\mathbb{R}^d$ based on $N$ i.i.d. samples. The reduced complexity MI estimator, called the ensemble dependency graph estimator (EDGE), combines randomized locality sensitive hashing (LSH), dependency graphs, and ensemble bias-reduction methods. We prove that EDGE achieves optimal computational complexity $O(N)$, and can achieve the optimal parametric MSE rate of $O(1/N)$ if the density is $d$ times differentiable. To the best of our knowledge EDGE is the first non-parametric MI estimator that can achieve parametric MSE rates with linear time complexity. We illustrate the utility of EDGE for the analysis of the information plane (IP) in deep learning. Using EDGE we shed light on a controversy on whether or not the  compression property of information bottleneck (IB) in fact holds for ReLu and other rectification functions in deep neural networks (DNN). 

\end{abstract}
\section{Introduction}

The Mutual Information (MI) is an often used measure of dependency between two random variables or vectors \cite{cover2012}, and it has a wide range of applications in information theory  \cite{cover2012} and machine learning \cite{class,hyvarinen2000independent}. Non-parametric MI estimation methods have been studied that use estimation strategies including KSG \cite{KSG}, KDE \cite{KDE} and Parzen window density estimation \cite{Parzen}. The performance of these estimators has been evaluated and compared based on both empirical studies \cite{khan} and asymptotic analysis \cite{kandasamy}. Recently several MI estimators have been proposed that can achieve parametric MSE rate of convergence. For example, in \cite{Poczos2014_2} a KDE plug-in estimator for R\'{e}nyi divergence and mutual information achieves the MSE rate of $O(1/N)$ when the densities are at least $d$ times differentiable. Another KDE based mutual information estimator was proposed in \cite{kandasamy} that can achieve the MSE rate of $O(1/N)$ when the densities are $d/2$ times differentiable. 
Recently Moon et al \cite{moon2017} and Gao et al \cite{Gao2017} respectively proposed KDE and KNN based MI estimators for random variables with mixtures of continuous and discrete components.
Most of these estimators, however, have high  computational cost and require knowledge of the density support boundary. 

In this paper we propose a reduced complexity MI estimator called the ensemble dependency graph estimator (EDGE). The estimator combines  randomized locality sensitive hashing (LSH), dependency graphs, and ensemble bias-reduction methods. A dependence graph is a bipartite directed graph consisting of two sets of nodes $V$ and $U$. The data points are mapped to the sets $V$ and $U$ using a randomized LSH function $H$ that depends on a hash parameter $\epsilon$. Each node is assigned a weight that is proportional to the number of hash collisions. Likewise, each edge between the vertices $v_i$ and $u_j$ has a weight proportional to the number of $(X_k,Y_k)$ pairs mapped to the node pairs $(v_i,u_j)$.
For a given value of the hash parameter $\epsilon$, a base estimator of MI is proposed as a weighted average of non-linearly transformed of the edge weights. The proposed EDGE estimator of MI is obtained by applying the method of weighted ensemble bias reduction \cite{Kevin16,moon2017} to a set of base estimators with different hash parameters.
This estimator is a non-trivial extension of the LSH divergence estimator defined in \cite{noshad_AISTAT}. LSH-based methods have previously been used for KNN search and graph constructions problems \cite{hash_KNN_graph, LSH_KNN}, and they result in fast and low complexity algorithms.

Recently, Shwartz-Ziv and Tishby utilized MI to study the training process in Deep Neural Networks (DNN) \cite{Tishby}. Let $X$, $T$ and $Y$ respectively denote the input, hidden and output layers. The authors of  \cite{Tishby} introduced the information bottleneck (IB) that represents the tradeoff between two mutual information measures: $I(X,T)$ and $I(T,Y)$. They observed that the training process of a DNN consists of two distinct phases; $1)$ an initial fitting phase in which $I(T,Y)$ increases, and $2)$ a subsequent compression phase in which $I(X,T)$ decreases. Saxe {\em et al} in \cite{saxe} countered the claim of \cite{Tishby}, asserting that this compression property is not universal, rather it depends on the specific activation function. Specifically, they claimed that the compression property does not hold for ReLu activation functions. The authors of \cite{Tishby} challenged these claims, arguing that the authors of \cite{saxe}  had not observed compression due to poor estimates of the MI. We use our proposed rate-optimal ensemble MI estimator to explore this controversy, observing that our estimator of MI does exhibit the compression phenomenon in the ReLU network studied by \cite{saxe}.  

Our contributions are as follows:

\begin{itemize}
\item 
To the best of our knowledge the proposed MI estimator is the first estimator to have 
 linear complexity and can achieve the optimal  MSE rate of $O(1/N)$.

\item 
The proposed MI estimator provides a simplified and unified treatment of mixed continuous-discrete variables. This is due to the hash function approach that is adopted.


\item EDGE is applied to IB theory of deep learning, and provides evidence that the compression property does indeed occur in ReLu DNNs, contrary to the claims of \cite{saxe}. 

\end{itemize}
The rest of the paper is organized as follows. In Section\ref{prob_def}, we introduce the general definition of MI and define the dependence graph. 
In Section \ref{Base_Estimator}, we introduce the hash based MI estimator and give theory for the bias and variance. In section \ref{EDGE_sec} we introduce the ensemble dependence graph MI estimator (EDGE) and show how the ensemble estimation method can be used to improve the convergence rates.
Finally, in Section \ref{experiments} we provide numerical results as well as study the IP in DNNs.\vspace{-0.4cm}


\section{Mutual Information}\label{prob_def}

In this section, we introduce the general mutual information function based on the f-divergence measure. Then, we define a consistent estimator for the mutual information function. 
Consider the probability measures $P$ and $Q$ on a Euclidean space $\mathcal{X}$. Let $g:(0,\infty)\to\mathbb{R}$ be a convex function with $g(1)=0$. The f-divergence between $P$ and $Q$ can be defined as follows \cite{Yuri2016, csiszar1995}.
\begin{align}\label{Div_def}
D(P\|Q):=\mathbb{E}_Q\of[g\of{\frac{dP}{dQ}}].
\end{align}

\begin{definition}[Mutual Information:]
Let $\mathcal{X}$ and $\mathcal{Y}$ be Euclidean spaces and let $P_{XY}$ be a probability measure on the space $\mathcal{X}\times\mathcal{Y}$. For any measurable sets $A\subseteq \mathcal{X}$ and $B\subseteq \mathcal{Y}$, we define the marginal probability measures $P_X(A):=P_{XY}(A\times \mathcal{Y})$ and $P_Y(B):=P_{XY}(\mathcal{X}\times B)$. Similar to \cite{Yuri2016, Gao2017}, the general MI denoted by $I(X,Y)$ is defined as
\begin{align} \label{def_g}
D(P_{XY}\|P_XP_Y)=\mathop{\mathbb{E}}_{P_XP_Y}\of[g\of{\frac{dP_{XY}}{dP_XP_Y}}],
\end{align}
where $\frac{dP_{XY}}{dP_XP_Y}$ is the Radon-Nikodym derivative, and $g:(0,\infty)\rightarrow \mathbb{R}$ is, as in \eqref{Div_def} a convex function with $g(1)=0$. Shannon mutual information is a particular cases of \eqref{Div_def} for which $g(x)=x\log x$. 
\end{definition}


\subsection{\textbf{Dependence Graphs}}

Consider $N$ i.i.d samples $(X_i,Y_i)$, $1\leq i\leq N$ drawn from the probability measure $P_{XY}$, defined on the space $\mathcal{X}\times \mathcal{Y}$. 
Define the sets $\mathbf{X}=\{X_1,X_2,...,X_N\}$ and $\mathbf{Y}=\{Y_1,Y_2,...,Y_N\}$.
The dependence graph $G(X,Y)$ is a directed bipartite graph, consisting of two sets of nodes $V$ and $U$ with cardinalities denoted as $|V|$ and $|U|$, and the set of edges $E_G$. Each point in the sets $\mathbf{X}$ and $\mathbf{Y}$ is mapped to the nodes in the sets $U$ and $V$, respectively, using the hash function $H$, described as follows.

\begin{figure}	
	\centering
	\includegraphics[width=0.9\columnwidth]{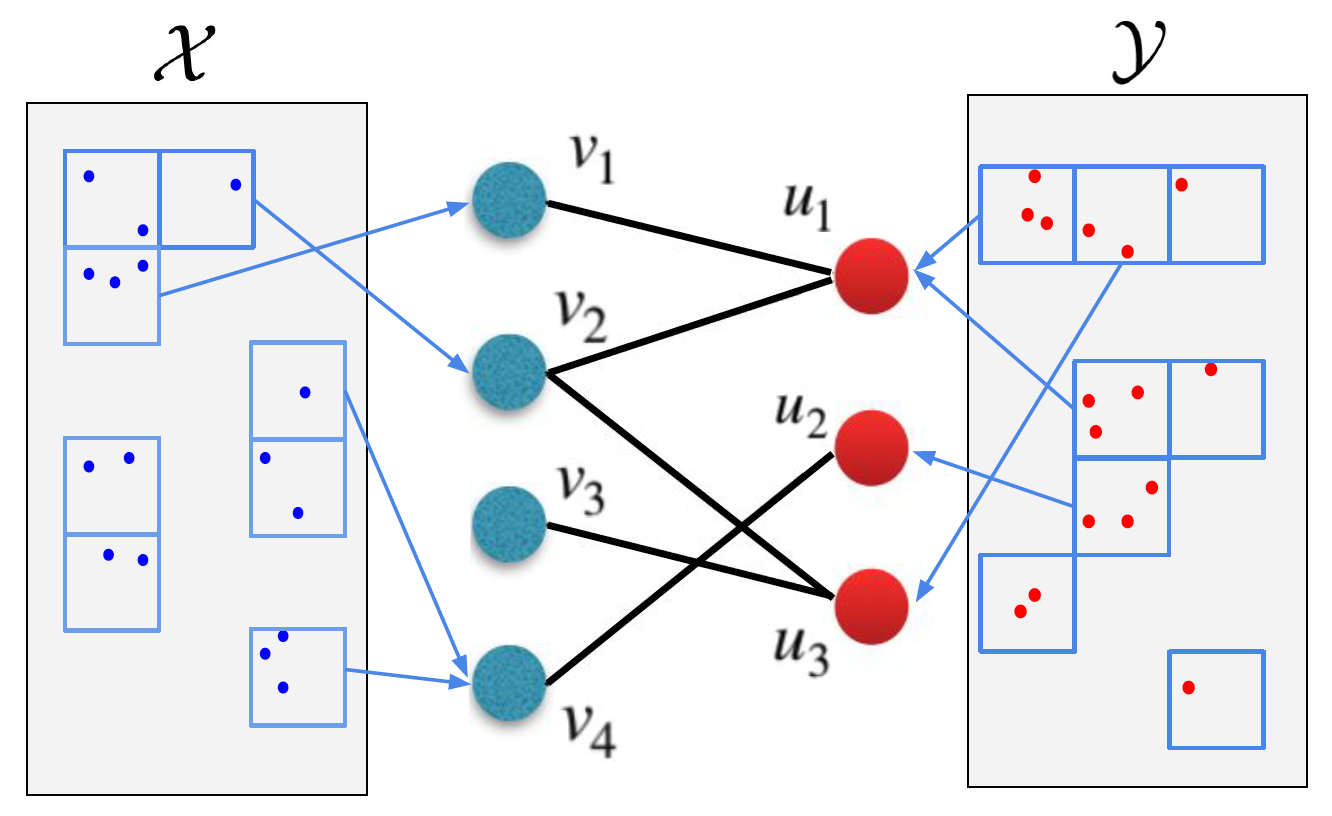}
	\vspace{-0.4cm} 
    \caption{Sample dependence graph with $4$ and $3$ respective distinct hash values of $\mathbf{X}$ and $\mathbf{Y}$ data jointly encoded with LSH, and the corresponding dependency edges.\vspace{-0.3cm} }
	\label{fig_graph}
\end{figure}

A vector valued hash function $H$ is defined in a similar way as defined in \cite{noshad_AISTAT}. First, define the vector valued hash function $H_1: \mathbb{R}^d\to \mathbb{Z}^d$ as 
\begin{align}\label{H1_def}
H_1(x)=\of[h_1(x_1), h_1(x_2), ... ,h_1(x_d)], 
\end{align}
where $x_i$ denotes the $i$th component of the vector $x$. In \eqref{H1_def}, each scalar hash function $h_1(x_i): \mathbb{R}\to \mathbb{Z}$ is given by
\begin{align}\label{def_eps}
h_1(x_i)=\floor{\frac{x_i+b}{\epsilon}},
\end{align}
for a fixed $\epsilon>0$, where $\lfloor y\rfloor$ denotes the floor function (the smallest integer value less than or equal to $y$), and $b$ is a fixed random variable in $[0,\epsilon]$. Let $\mathcal{F}:=\{1,2,..,F\}$, where $F:=c_HN$ and $c_H$ is a fixed tunable integer. We define a random hash function $H_2:\mathbb{Z}^d\to \mathcal{F}$ with a uniform density on the output and consider the combined hashing function
\begin{align}\label{Hash_def}
H(x):=H_2(H_1(x)),
\end{align}
which maps the points in $\mathbb{R}^d$ to $\mathcal{F}$.

$H(x)$ reveals the index of the mapped vertex in $G(X,Y)$. The weights $\omega_i$ and $\omega'_j$ corresponding to the nodes $v_i$ and $u_j$, and $\omega_{ij}$, the weight of the edge $(v_i,u_j)$, are defined as follows.
\begin{align}
\omega_i=\frac{N_i}{N},\hspace{0.1\linewidth}
\omega'_j=\frac{M_j}{N},\hspace{0.1\linewidth}
\omega_{ij}=\frac{N_{ij}N}{N_iM_j},
\end{align}
where $N_i$ and $M_j$ respectively are the the number of hash collisions at the vertices $v_i$ and $u_j$, and $N_{ij}$ is the number of joint collisions of the nodes $(X_k,Y_k)$ at the vertex pairs $(v_i,u_j)$. The number of hash collisions is defined as the number of instances of the input variables map to the same output value. In particular, 
\begin{align}
N_{ij}:=\#\{(X_k,Y_k) \text{ s.t  } H(X_k)=i \text{ and  } H(Y_k)=j \}.
\end{align}
Fig. \ref{fig_graph} represents a sample dependence graph. Note that the nodes and edges with zero collisions do not show up in the dependence graph. 

\vspace{-0.4cm}
\section{The Base Estimator of MI}
\vspace{-0.3cm}
\label{Base_Estimator}
\subsection{\textbf{Assumptions}}
\vspace{-0.2cm}
\begin{definition}[]
The following are the assumptions we make on the probability measures and $g$:

\textbf{A1. } The support sets $\mathcal{X}$ and $\mathcal{Y}$ are bounded.

\textbf{A2. } The following supremum exists and is bounded:\vspace{-0.1cm} $$\sup_{P_XP_Y}g\of{\frac{dP_{XY}}{dP_XP_Y}}\leq U.\vspace{-0.1cm}$$

\textbf{A3.} Let $x_D$ and $x_C$ respectively denote the discrete and continuous components of the vector $x$. Also let $f_{X_C}(x_C)$ and $p_{X_D}(x_D)$ respectively denote density and pmf functions of these components associated with the  probability measure $P_X$. The density functions $f_{X_C}(x_C)$, $f_{Y_C}(y_C)$, $f_{X_CY_C}(x_C,y_C)$, and the conditional densities $f_{X_C|X_D}(x_C|x_D)$, $f_{Y_C|Y_D}(y_C|y_D)$, $f_{X_CY_C|X_DY_D}(x_C,y_C|x_D,y_D)$ are H\"{o}lder continuous.

\end{definition}

\begin{definition}[H\"{o}lder continuous functions: ]\label{Holder}
Given a support set $\mathcal{X}$, a function $f:\mathcal{X} \to \mathbb{R}$ is called H\"{o}lder continuous with parameter $0<\gamma\leq 1$, if there exists a positive constant $G_f$, possibly depending on $f$, such that for every $x\neq y \in \mathcal{X}$,
\begin{equation}\label{Holder_eq}
|f(y)-f(x)|\leq G_f\|y-x\|^{\gamma}.
\end{equation}
\end{definition}

\textbf{A4. } Assume that the function $g$ in \eqref{def_g} is Lipschitz continuous; i.e. $g$ is H\"{o}lder continuous with $\gamma=1$.
\vspace{-0.3cm}
\subsection{\textbf{The Base Estimator of MI}}
For a fixed value of the hash parameter $\epsilon$, we propose the following base estimator of MI \eqref{def_g} function based on the dependence graph:
\begin{equation}\label{est_def}
\widehat{I}(X,Y):=\sum_{e_{ij}\in E_G} \omega_{i}\omega'_{j}\widetilde{g}\of{\omega_{ij}},
\end{equation}
where the summation is over all edges $e_{ij}:(v_i\to u_j)$ of $G(X,Y)$ having non-zero weight and $\widetilde{g}(x):=\max\{g(x),U\}$. 

When $X$ and $Y$ are strongly dependent, each point $X_k$ hashed into the bucket (vertex) $v_i$ corresponds to a unique hash value for $Y_k$ in $U$. Therefore, asymptotically $\omega_{ij}\to 1$ and the mutual information estimation in \eqref{est_def} takes its maximum value. On the other hand, when $X$ and $Y$ are independent, each point $X_k$ hashed into the bucket (vertex) $v_i$ may be associated with different values of $Y_k$, and therefore asymptotically $\omega_{ij}\to \omega_j$ and the Shannon MI estimation tends to $0$.


\subsection{Various LSH Functions}
There are various types of LSH functions \cite{SimHash, p_stable, LSH_survey}, and all of them share the common property that they map similar items to the same bins with high probability.

In equations \eqref{H1_def} and \eqref{def_eps} we considered a simple floor function on the scaled input, however in general, any other type of LSH might be used for our estimation method. In particular, the hash functions based on random projections can reduce the dimensionality of data. SimHash \cite{SimHash}, which is based on cosine distance, and the LSH based on p-stable distributions \cite{p_stable} are among well known LSH functions that reduce the dimension of data. For example, the LSH based on p-stable
distribution is defined similarly to the floor hash function in \eqref{H1_def} and \eqref{def_eps}, except that the input vector is projected on random hyperplanes with p-stable distributions. The formal definition is $H_{p-stable}:\mathbb{R}^d\to \mathbb{Z}^r$,

\begin{align}
    H_{p-stable}(x)= H_1(XW),
\end{align}
where $H_1$ is defined in \eqref{H1_def}, and $W$ is a $d\times r$ matrix with entries chosen independently from a stable distribution. For high-dimensional datasets one can choose $r<<d$ in order to reduce the  dimensionality. Finally, note that for theoretical analysis, we only focus on performance of the simple floor hash function defined in \eqref{H1_def} and \eqref{def_eps}.

\vspace{-0.2cm}
\subsection{\textbf{Convergence Rates} }\label{results}
In the following theorems we state upper bounds on the bias and variance rates of the proposed MI estimator \eqref{est_def}. The proofs are given in appendices A and B.
We define the notations $\mathbb{B}[\hat{T}]=\mathbb{E}[\hat{T}]-T$ for bias and $\mathbb{V}[\hat{T}]=\mathbb{E}[\hat{T}^2]-\mathbb{E}[\hat{T}]^2$ for variance of $\hat{T}$. The following theorem states an upper bound on the bias.

\begin{theorem} \label{bias_theorem}
Let $d=d_X+d_Y$ be the dimension of the joint random variable $(X,Y)$. Under the aforementioned assumptions \textbf{A1-A4}, and assuming that the density functions in \textbf{A3} have bounded derivatives up to order $q\geq 0$, the following upper bound on the bias of the estimator in \eqref{est_def} holds
\begin{align}\label{bias_terms}
\mathbb{B}\of[\widehat{I}(X,Y)] = 
     \begin{cases}
       O\of{\epsilon^\gamma}+O\of{\frac{1}{N\epsilon^d}}, &\quad q=0\vspace{2mm}\\
       \sum_{i=1}^q C_i\epsilon^{i}+O\of{\epsilon^q}+O\of{\frac{1}{N\epsilon^d}} &\quad q\geq 1,
     \end{cases}
     \vspace{-0.5cm}
\end{align}
where  $\epsilon$ is the hash parameter in \eqref{def_eps}, $\gamma$ is the smoothness parameter in \eqref{Holder_eq}, and $C_i$ are real constants.

\end{theorem}

In \eqref{bias_terms}, the hash parameter, $\epsilon$ needs to be a function of $N$ to ensure that the bias converges to zero. For the case of $q=0$, the optimum bias is achieved when $\epsilon=\of{\frac{1}{N}}^{\gamma/(\gamma+d)}$. When $q\geq 1$, the optimum bias is achieved for $\epsilon=\of{\frac{1}{N}}^{1/(1+d)}$.

\begin{theorem}\label{variance}
Under the assumptions \textbf{A1-A4} the variance of the proposed estimator can be bounded as $\mathbb{V}\of[\widehat{I}(X,Y)]\leq O\of{\frac{1}{N}}$.
Further, the variance of the variable $\omega_{ij}$ is also upper bounded by $O(1/N)$.
\end{theorem}


\subsection{Computational Complexity}

We analyze the computational complexity of the proposed estimator. The procedure for estimation of MI equivalent to the proposed estimator in \eqref{est_def} is given in Algorithm \ref{algo}.

\begin{algorithm} \label{algo}
\DontPrintSemicolon
\SetKwInOut{Input}{Input}\SetKwInOut{Output}{Output}
\Input{$N$ i.i.d samples $(X_k,Y_k)$, $1\leq k\leq N$.}

\BlankLine

\For {each $k\in 1:N$}{
		$i\leftarrow H(\mathbf{X_k})$\\
		$j\leftarrow H(\mathbf{Y_k})$\\
		 $N_i\leftarrow N_i+1$\\
         $M_j\leftarrow M_j+1$\\
         $N_{ij}\leftarrow N_{ij}+1$
		 }
$\omega_i\leftarrow N_i/N; \omega'_j\leftarrow M_j/N; \omega_{ij}\leftarrow N_{ij}N/N_iM_j;$
$\widehat{I}\leftarrow\sum_{e_{ij}} \omega_{i}\omega'_{j}\widetilde{g}\of{\omega_{ij}}$

\Output{$\widehat{I}$}

\caption{MI Dependence Graph Estimator}
\end{algorithm}

We go over all of the data points and map them using the hash function $H$. We compute the number of marginal and joint collisions ($N_i,M'_j, N_{ij}$) and based on these we compute the vertex and edge weights. Note that computing the hashing of all of the data points takes about $O(N)$ time. Finally in the last line we compute the MI estimate by going over all of the edges $e_{ij}$. The number the edges is upper bounded by $O(N)$, since each edge correspond to at least one pair of $(X_i,Y_i)$. Finally, note that for high-dimensional data sets, the computational of computing the hash function of each input may depend on $d$, however, is would not be greater than $O(d)$. Hence, overall the the computational complexity of computing the proposed MI estimate is linear with respect to both $N$ and $d$.

\subsection{Comparison to the Other Estimation Methods}

So far, various estimation methods for information measures (entropy, divergence and mutual information) have been proposed, most of which are based on kernel density estimates (KDE) \cite{Kevin16}, k-nearest neighbors (KNN) \cite{Noshad2017} or histogram binning \cite{histogram}. Certain KDE and KNN based estimators can achieve the optimal parametric MSE rate \cite{Kevin16,Noshad2017,krishnamurthy2014}, however, implementation of the KDE and KNN methods respectively have $O(N^2)$ and $O(kN\log N)$ computational complexity, where $N$ is the number of samples. One could probably could approximation methods for KDE and KNN, however, there would be no theoretical guarantees for the estimation based on these approximations. Empirical histograms, on the other hand, are simpler and easier to implement, however, their convergence rate is not as good as the KNN and KDE based methods \cite{histogram}. 

LSH methods have previously been applied to the approximate nearest neighbor search, however, it has never been directly utilized for estimation of densities or information theoretic quantities.  
The simplest LSH function considered in \eqref{def_eps} has similarities with the histogram estimator in terms of binning, however, there are also certain differences. As opposed to the LSH based method, the histogram binning requires a pre-knowledge of the support set or needs an extra computation to estimate the support set. The number of 
the bins in the histogram estimator gets exponentially large with increasing dimension which results in a huge computational complexity for high-dimensional datasets. In addition, most of the bins would be empty. In contrast, the LSH based method results in no empty hash bins, the number of the bins is upper bounded by $O(N)$, and the computational complexity is linear in dimension and the number of samples. The plug-in methods including histograms, KDE and KNN require the estimation of the densities $p_X$, $p_Y$ and $p_{XY}$ for computation of mutual information, while the proposed LSH-based method finds the mapping of $\X$ and $\Y$ data points, and then estimate the mutual information, based on the hash collisions. Finally, by giving an accurate bias and variance rates for the base LSH estimator, we use an ensemble estimation technique to achieve the optimal parametric convergence rate, discussed in the following section.

\section{Ensemble Dependence Graph Estimator (EDGE)}\label{EDGE_sec}
\vspace{-0.3cm}
\begin{definition}[]
Given the expression for the bias in Theorem \ref{bias_theorem}, the ensemble estimation technique proposed in \cite{Kevin16} can be applied to improve the convergence rate of the MI estimator \eqref{est_def}. Assume that the densities in \textbf{A3} have continuous bounded derivatives up to the order $q$, where $q\geq d$.
Let $\mathcal{T}:=\{t_1,...,t_T\}$ be a set of index values with $t_i<c$, where $c>0$ is a constant. Let $\epsilon(t):=tN^{-1/2d}$. For a given set of weights $w(t)$ the weighted ensemble estimator is then defined as 
\begin{align}\label{EDGE_def}
\widehat{I}_w:=\sum_{t\in \mathcal{T}}w(t)\widehat{I}_{\epsilon(t)},
\end{align}
where $\widehat{I}_{\epsilon(t)}$ is the mutual information estimator with the parameter $\epsilon(t)$. Using \eqref{bias_terms}, for $q>0$ the bias of the weighted ensemble estimator \eqref{EDGE_def} takes the form
\begin{equation}
\mathbb{B}(\hat{I}_w) = \sum_{i=1}^q Ci N^{-\frac{i}{2d}} \sum_{t\in \mathcal{T}} w(t) t^{i} +O\of{{\frac{t^d}{N^{1/2}}}}+O\of{\frac{1}{N\epsilon^d}}
\label{Ensemble_bias}
\end{equation}

Given the form (\ref{Ensemble_bias}), as long as $T\geq q$, we can select the  weights $w(t)$ to force to zero the slowly decaying terms in (\ref{Ensemble_bias}), i.e.  $\sum_{t\in \tau} w(t)t^{i/d}=0$ subject to the constraint that$\sum_{t\in \tau} w(t)=1$. However, $T$ should be strictly greater than $q$ in order to control the variance, which is upper bounded by the euclidean norm squared of the weights  $\omega$. In particular we have the following theorem (the proof is given in Appendix C):

\begin{theorem} \label{ensemble_theorem}
For $T>d$ let  $w_0$ be the solution to:
\begin{align}
\min_w &\qquad \|w\|_2 \nonumber\\
\textit{subject to} &\qquad \sum_{t\in \mathcal{T}}w(t)=1, \nonumber\\
&\qquad \sum_{t\in \mathcal{T}}w(t)t^{i}=0, i\in \mathbb{N}, i\leq d.
\end{align}
Then the MSE rate of the ensemble estimator $\widehat{I}_{w_0}$ is $O(1/N)$.
\end{theorem}

\end{definition}

\section{Experiments}\label{experiments}

We first use simulated data to compare the proposed estimator to the competing MI estimators Ensemble KDE (EKDE)\cite{moon2017}, and generalized KSG \cite{Gao2017}. 
Both of these estimators work on mixed continuous-discrete variables. We also apply EDGE to study the information bottleneck in different networks trained on MNIST hand-written datset.

Fig. \ref{figure1}, shows the MSE estimation rate of Shannon MI between the continuous random variables $X$ and $Y$ having the  relation $Y=X+a N_U$,
where $X$ is a $2D$ Gaussian random variable with the mean $[0, 0]$ and covariance matrix $C=I_2$. Here $I_d$ denote the $d$-dimensional identity matrix. $N_U$ is a uniform random vector with the support $\mathcal{N}_U=[0,1]\times [0,1]$.  
We compute the MSE of each estimator for different sample sizes. The MSE rates of EDGE, EKDE and KSG are compared for $a=1/5$. Further, the MSE rate of EDGE is investigated for noise levels of $a=\{1/10,1/5,1/2,1\}$. As the dependency between $X$ and $Y$ increases the MSE rate becomes slower. 


\begin{figure}	
	\centering
	\includegraphics[width=1\columnwidth]{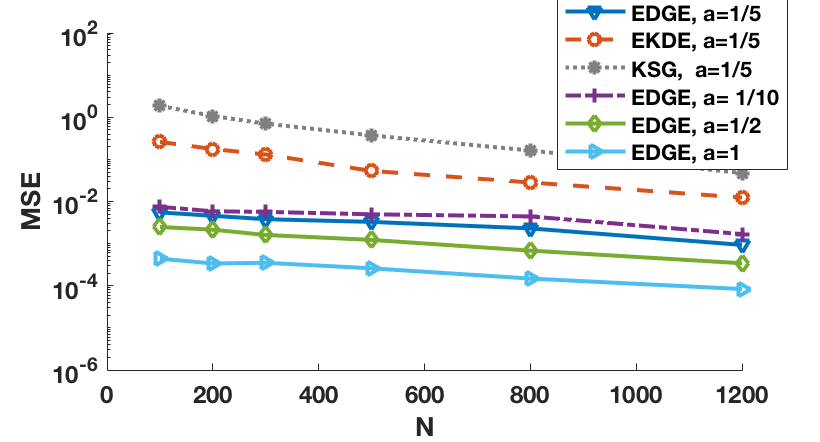}
	\vspace{-0.6cm}
    \caption{MSE comparison of EDGE, EDKE and KSG Shannon MI estimators. $X$ is a $2D$ Gaussian random variable with unit covariance matrix. $Y=X+aN_U$, where $N_U$ is a uniform noise. The MSE rates of EDGE, EKDE and KSG are compared for various values of $a$.}
	\label{figure1}
    \vspace{-1.2em}
\end{figure}

Fig. \ref{figure2}, shows the MSE estimation rate of Shannon MI between a discrete random variables $X$ and a continuous random variable $Y$.
We have $X\in \{1,2,3,4\}$, and each $X=x$ is associated with multivariate Gaussian random vector $Y$, with $d=4$, the expectation $[x/2, 0, 0, 0]$ and covariance matrix $C=I_4$. 
In general in Figures \ref{figure1} and \ref{figure2}, EDGE has better convergence rate than EKDE and KSG estimators. Fig. \ref{fig:runtime} represents the runtime comparison for the same experiment as in Fig. \ref{figure2}. It can be seen from this graph how fast our proposed estimator performs compared to the other other methods.

\begin{figure}	
	\centering
	\includegraphics[width=1\columnwidth]{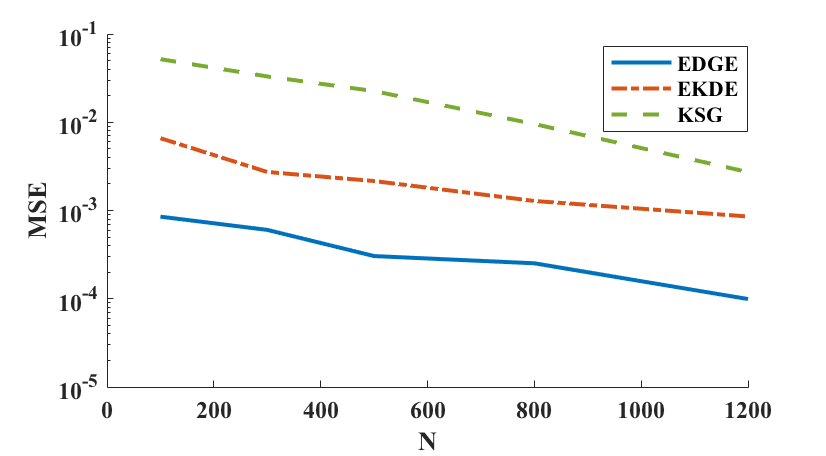}
    \caption{MSE comparison of EDGE, EDKE and KSG Shannon MI estimators. $X\in \{1,2,3,4\}$, and each $X=x$ is associated with multivariate Gaussian random vector $Y$, with $d=4$, the mean $[x/2, 0, 0, 0]$ and covariance matrix $C=I_4$.}
    \vspace{-1.5em}
    \label{figure2}
\end{figure}

\begin{figure}	
	\centering
	\includegraphics[width=1\columnwidth]{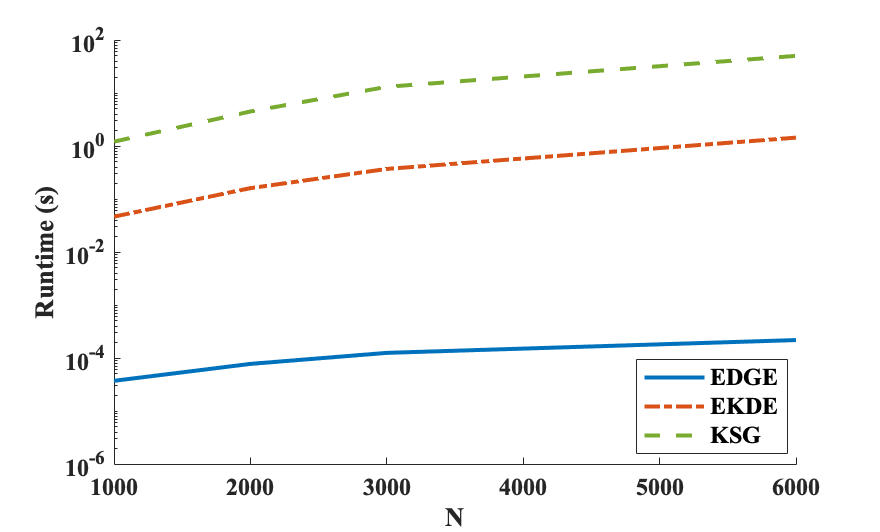}
    \caption{Runtime comparison of EDGE, EDKE and KSG Shannon MI estimators. $X\in \{1,2,3,4\}$, and each $X=x$ is associated with multivariate Gaussian random vector $Y$, with $d=4$, the mean $[x/2, 0, 0, 0]$ and covariance matrix $C=I_4$.}
    \vspace{-1.5em}
	\label{fig:runtime}
\end{figure}


Next, we use EDGE to study the information bottleneck \cite{Tishby} in DNNs. Fig. \ref{IP} represents the information plane of a DNN with $4$ fully connected hidden layers of width $784-1024-20-20-20-10$ with tanh and ReLU activations. The sequence of colored points shows different iterations of the training process. Each gray line connects the points with the same iterations for diferent layers. The left most sequence of points corresponds to the last hidden layer and the right most sequence of points corresponds to the first hidden layer.
The network is trained with Adam optimization with a learning rate of $0.003$ and cross-entropy loss functions to classify the MNIST handwritten-digits dataset. 
We repeat the experiment for $20$ iterations with different randomized initializations and take the average over all experiments. In both cases of ReLU and tanh activations we observe some degree of compression in all of the hidden layers. However, the amount of compressions is different for ReLU and tanh activations. The average test accuracy in both of these networks are around $0.98$. This network is the same as the one studied in \cite{saxe}, for which it is claimed that no compression happens with a ReLU activation. The base estimator used in \cite{saxe} provides KDE-based lower and upper bounds on the true MI \cite{kolchinsky2017}. According to our experiments (not shown) the upper bound is in some cases twice as large as the lower bound. In contrast, our proposed ensemble method estimates the exact mutual information with significantly higher accuracy. 



\begin{figure}	
	\centering
	\includegraphics[width=1\columnwidth]{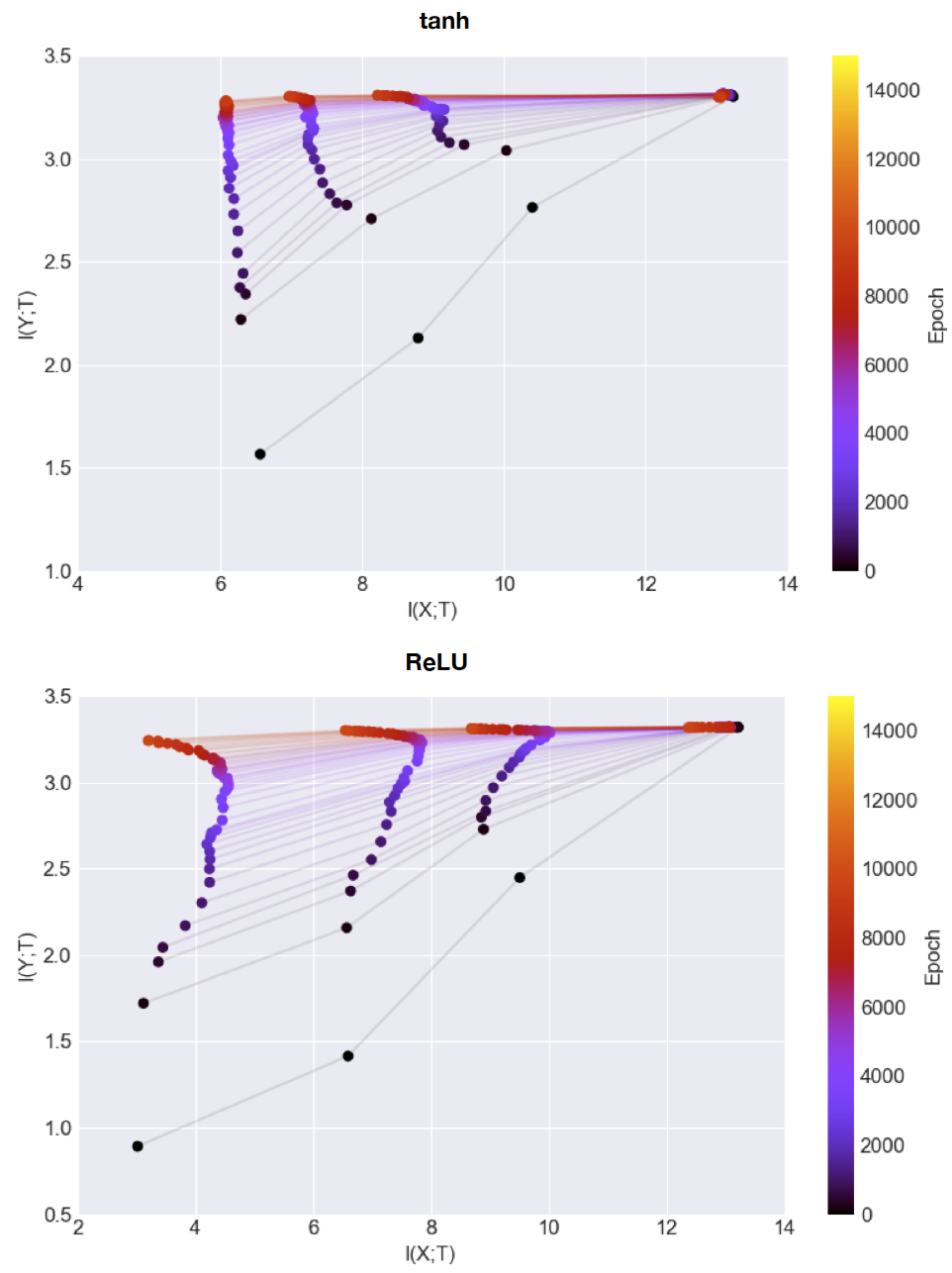}
    \caption{Information plane estimated using EDGE for a neural network of size $784-1024-20-20-20-10$ trained on the MNIST dataset with tanh (top) and ReLU (bottom) activations.}
	\label{IP}
    \vspace{-1.2em}
\end{figure}



Fig. \ref{200_100_60_30_relu} represents the information plane for another network with $4$ fully connected hidden layers of width $784-200-100-60-30-10$ with ReLU activation. The network is trained with Adam optimization with a learning rate of $0.003$ and cross-entropy loss functions to classify the MNIST handwritten-digits dataset. Again, we observe compression for this network with ReLU activation.

\begin{figure}	
	\centering
	\includegraphics[width=1\columnwidth]{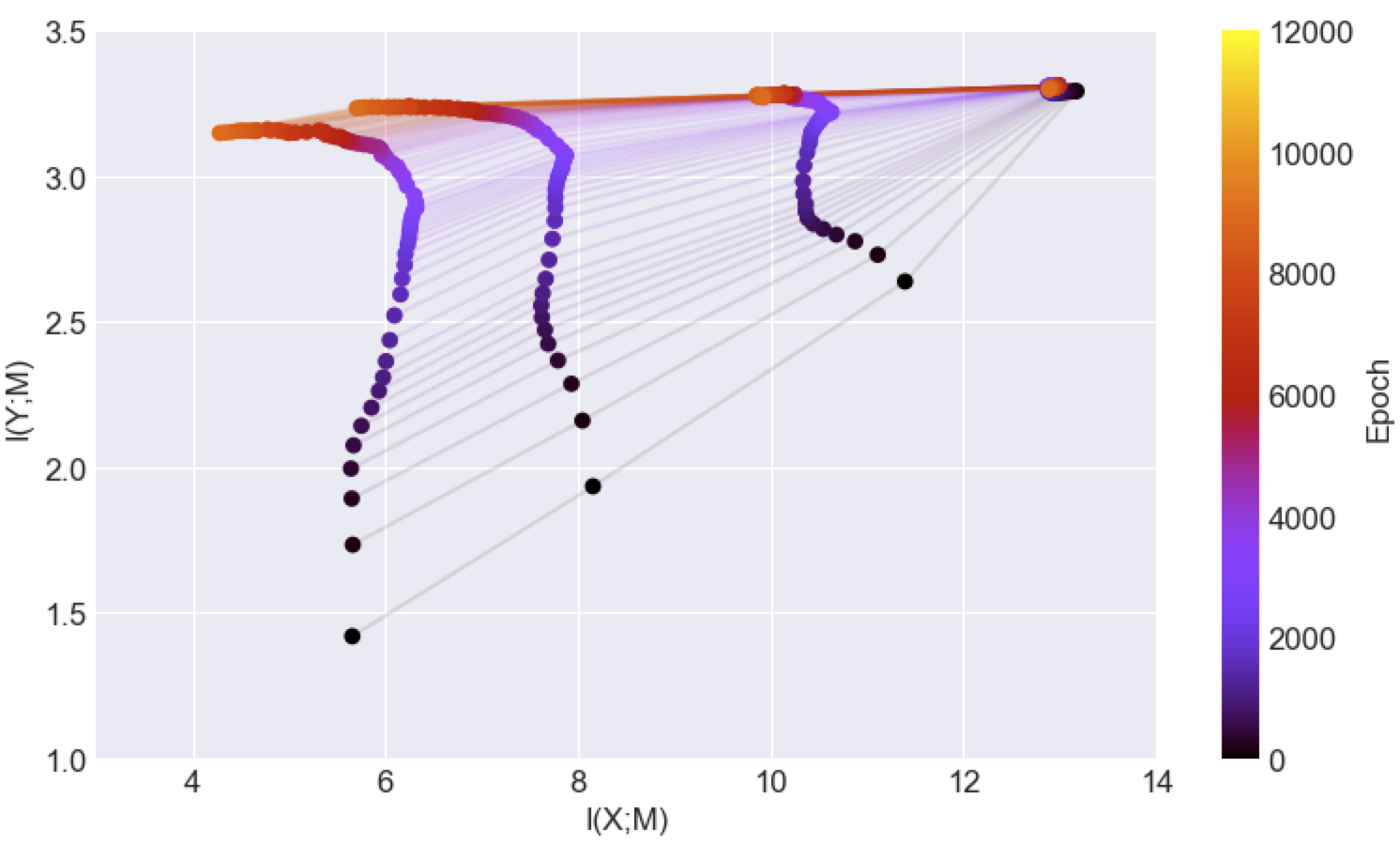}
    \caption{Information plane estimated using EDGE for a neural network of size $784-200-100-60-30-10$ trained on the MNIST dataset with ReLU activation.}
	\label{200_100_60_30_relu}
    \vspace{-1.2em}
\end{figure}


Finally, we study the information plane curves in a CNN with three convolutioal ReLU layers and a dense ReLU layer. The convolutional layers respectively have depths of $4,8,16$ and  the dense layer has the dimension $256$. 
Max-pooling functions are used in the second and third layers. Note although for a certain initialization of the weights this model can achieve the test accuracy of $0.99$, the average test accuracy (over different weight initializations) is around $0.95$. That's why the converged point of the last layer has smaller $I(T,Y)$ compared to the examples in Fig. \ref{IP}, which achieves the average test accuracy of $0.98$. Another interesting point about the information plane in CNN is that the convolutional layers have larger $I(T,Y)$ compared to the hidden layers in the fully connected models in \ref{IP} and \ref{200_100_60_30_relu}, which implies that the convolutional layers can extract almost all of the useful information about the labels after small number of iterations.


\begin{figure}	
	\centering
	\includegraphics[width=1\columnwidth]{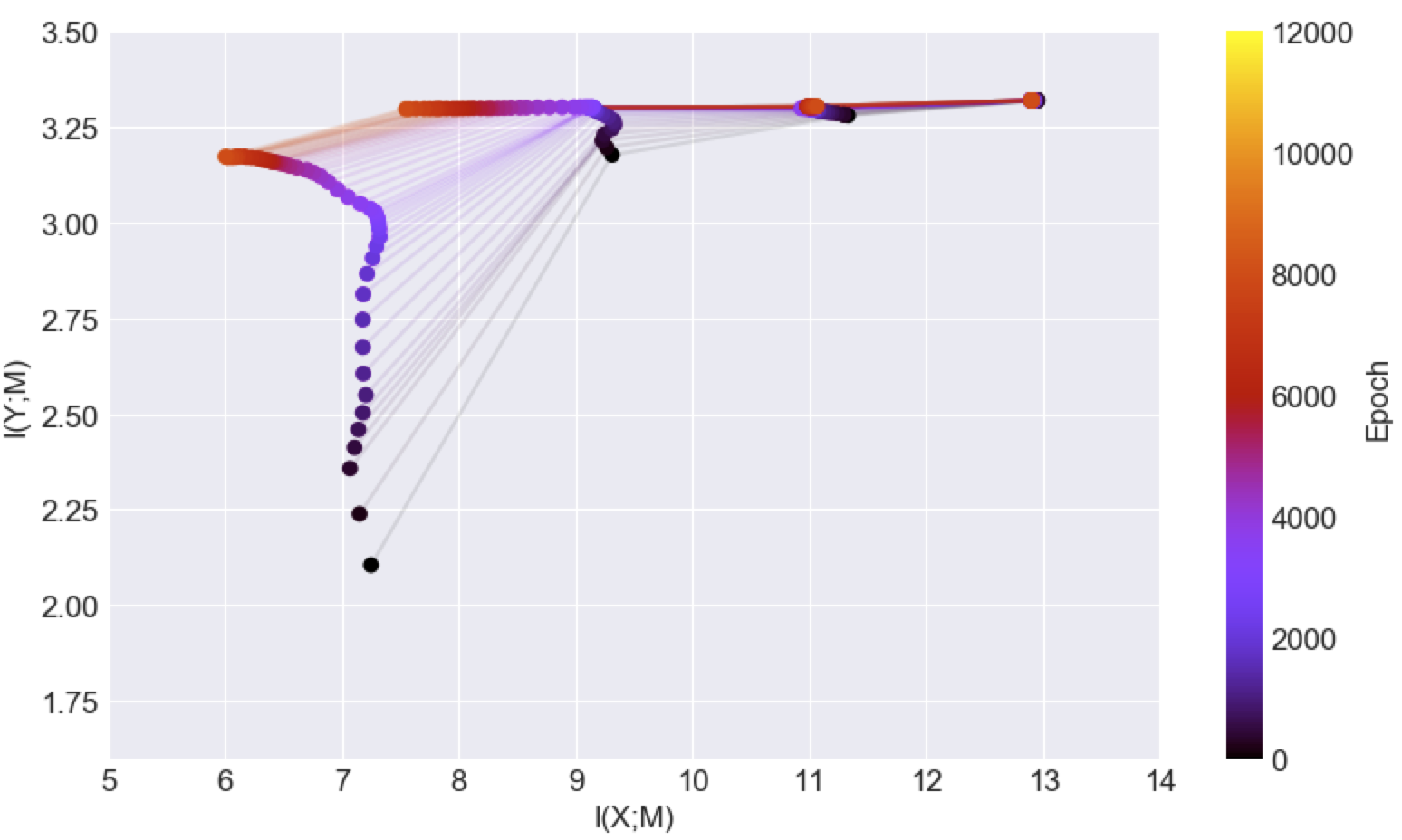}
    \caption{Information plane estimated using EDGE for a CNN consisting of three convolutioal ReLU layers with the respective depths of $4,8,16$ and a dense ReLU layer with the size of $256$.}
	\label{CNN}
    \vspace{-1.2em}
\end{figure}

\vspace{-0.2cm}
\section{Conclusion}
\vspace{-0.25cm}
In this paper we proposed a fast non-parametric estimation method for MI based on random hashing, dependence graphs, and ensemble estimation. Remarkably, the proposed estimator has linear computational complexity and attains optimal (parametric) rates of MSE convergence. We provided bias and variance convergence rate, and validated our results by numerical experiments.


\bibliographystyle{ieeetr}
\bibliography{citations.bib}

\onecolumn
\newpage



\section*{A. Bias Proof}
We first prove a theorem that establishes an upper bound on the number of vertices in $V$ and $U$.

\begin{lemma}\label{Lemma_num_ver}
Cardinality of the sets $U$ and $V$ are upper bounded as $|V|\leq O\of{\epsilon^{-d}}$ and $|U|\leq O\of{\epsilon^{-d}}$, respectively.
\end{lemma}

\begin{proof}

Let $\{\widetilde{X}_i\}_{i=1}^{L_X}$ and $\{\widetilde{Y}_i\}_{i=1}^{L_Y}$\vspace{2mm} respectively denote distinct outputs of $H_1$ with the $N$ i.i.d points $X_k$ and $Y_k$ as input. Then according to \cite{noshad_AISTAT} (Lemma 4.1), we have
\begin{align}\label{bound_LX_LY}
L_X\leq O\of{\epsilon^{-d}}, \hspace{4mm} L_Y\leq O\of{\epsilon^{-d}}.
\end{align}

Simply, because of the deterministic feature of $H_2$, the number of its distinct inputs is greater than or equal to the number of its outputs. So, $|V|\leq L_X$ and $|U|\leq L_Y$. Using the bounds in \eqref{bound_LX_LY} completes the proof.  
\end{proof}

The bias proof is based on analyzing the hash function defined in \eqref{Hash_def}. The proof consists of two main steps: 1) Finding the expectation of hash collisions of $H_1$; and 2) Analyzing the collision error of $H_2$. An important point about $H_1$ and $H_2$ is that collision of $H_1$ plays a crucial role in our estimator, while the collision of $H_2$ adds extra bias to the estimator. We introduce the following events to formally define these two biases:

\begin{align}
E_{ij}: & \text{The event that there is an edge between the vertices $v_i$ and $u_j$}.\nonumber\\ 
E_{\mathcal{E}}: & \text{The event that $\mathcal{E}$ is the set of all edges in $G$, i.e. $\mathcal{E}=E_G$}.\nonumber\\
E_{v_i}^{>0}: & \text{The event that there is at least one vector from $\{\widetilde{X}_i\}_{i=1}^{L_X}$ that maps to $v_i$ using $H_2$ }.\nonumber\\
E_{v_i}^{=1}: & \text{The event that there is exactly one vector from $\{\widetilde{X}_i\}_{i=1}^{L_X}$ that maps to $v_i$ using $H_2$ }.\nonumber\\
E_{v_i}^{>1}: & \text{The event that there are at least two vectors from $\{\widetilde{X}_i\}_{i=1}^{L_X}$ that map to $v_i$ using $H_2$ }.
\end{align}
$E_{u_i}^{>0}$, $E_{u_i}^{=1}$ and $E_{u_i}^{>1}$ are defined similarly. Further, let for any event $E$, $\overline{E}$ denote the complementary event. Let $E_{ij}^{=1}:= E_{v_i}^{=1}\cap E_{u_i}^{=1}$. Finally, we define $E^{=1}:=\of{\cap_{i=1}^{L_X}E_{v_i}^{=1}}\cap\of{\cap_{j=1}^{L_Y}E_{u_j}^{=1}}$, which represent the event of no collision.

Consider the notation $\widetilde{I}(X,Y):=\sum_{e_{ij}\in E_G} \omega_{i}\omega'_{j}\widetilde{g}\of{\omega_{ij}}$ (Notice the difference from the definition in \eqref{est_def}). We can derive its expectation as

\begin{align}\label{bias_proof_1}
\E{\widetilde{I}(X,Y)}&=\E{\sum_{e_{ij}\in E_G} \omega_{i}\omega'_{j}\widetilde{g}\of{\omega_{ij}}\bigg\vert E_G}\nonumber\\
&=\sum_{e_{ij}\in E_G} \E{\omega_{i}\omega'_{j}\widetilde{g}\of{\omega_{ij}}\vert E_{ij}}\nonumber\\
&=\sum_{e_{ij}\in E_G}P(E_{ij}^{=1}\vert E_{ij}) \E{\omega_{i}\omega'_{j}\widetilde{g}\of{\omega_{ij}}\middle\vert E_{ij}^{=1},E_{ij}}\nonumber\\
&+\sum_{e_{ij}\in E_G}P\of{\overline{E_{ij}^{=1}}\vert E_{ij}} \E{\omega_{i}\omega'_{j}\widetilde{g}\of{\omega_{ij}}\middle\vert \overline{E_{ij}^{=1}},E_{ij}}.
\end{align}
Note that the second term in \eqref{bias_proof_1} is the bias due to collision of $H_2$ and we denote this term by $\mathbb{B}_H$. 

\subsection{Bias Due to Collision}

The following lemma states an upper bound on the bias error caused by $H_2$.
\begin{lemma}\label{collision_lemma}
The bias error due to collision of $H_2$ is upper bounded as
\begin{align}
\mathbb{B}_H \leq O\of{\frac{1}{\epsilon^{d}N}}.
\end{align}
\end{lemma}

Before proving this lemma, we provide the following lemma.

\begin{lemma}\label{X_ij}
$P(E_{ij}^{=1}\vert E_{ij})$ is given by
\begin{align}\label{E_ij_eq}
P(E_{ij}^{=1}\vert E_{ij})= 1-O\of{\frac{1}{\epsilon^dN}}.
\end{align}
\end{lemma}

\begin{proof}

Let $\widetilde{\mathbf{X}}=\widetilde{\mathbf{x}}$ and $\widetilde{\mathbf{Y}}=\widetilde{\mathbf{y}}$ respectively abbreviate the equations 
$\widetilde{X}_1=\widetilde{x}_1,...,\widetilde{X}_{L_X}=\widetilde{x}_{L_X}$ and $\widetilde{Y}_1=\widetilde{y}_1,...,\widetilde{Y}_{L_Y}=\widetilde{y}_{L_Y}$. Let $\widetilde{\mathbf{x}}:=\{\widetilde{x}_1,\widetilde{x}_2,...,\widetilde{x}_{L_X}\}$ and $\widetilde{\mathbf{y}}:=\{\widetilde{y}_1,\widetilde{y}_2,...,\widetilde{y}_{L_Y}\}$. Define $\widetilde{\mathbf{z}}:=\widetilde{\mathbf{x}}\cup \widetilde{\mathbf{y}}$ and $L_Z:=|\widetilde{\mathbf{z}}|$.

\begin{align}
P(E_{ij}^{=1}\vert E_{ij})&=\sum_{\widetilde{\mathbf{x}},\widetilde{\mathbf{y}}} P\of{\widetilde{\mathbf{X}}=\widetilde{\mathbf{x}},\widetilde{\mathbf{Y}}=\widetilde{\mathbf{y}}\vert E_{ij}}P(E_{ij}^{=1}\vert E_{ij}, \widetilde{\mathbf{X}}=\widetilde{\mathbf{x}},\widetilde{\mathbf{Y}}=\widetilde{\mathbf{y}}).
\end{align}

Define $a=2$ for the case $i\neq j$ and $a=1$ for the case $i=j$. Then we have

\begin{align}
P(E_{ij}^{=1}\vert E_{ij})&=\sum_{\widetilde{\mathbf{x}},\widetilde{\mathbf{y}}} P\of{\widetilde{\mathbf{X}}=\widetilde{\mathbf{x}},\widetilde{\mathbf{Y}}=\widetilde{\mathbf{y}}\vert E_{ij}}O\of{\of{\frac{F-a}{F}}^{L_Z-a}}\nonumber\\
&=\sum_{\widetilde{\mathbf{x}},\widetilde{\mathbf{y}}} P\of{\widetilde{\mathbf{X}}=\widetilde{\mathbf{x}},\widetilde{\mathbf{Y}}=\widetilde{\mathbf{y}}\vert E_{ij}}\of{1-O\of{\frac{L_Z}{F}}}\nonumber\\
&\leq\sum_{\widetilde{\mathbf{x}},\widetilde{\mathbf{y}}} P\of{\widetilde{\mathbf{X}}=\widetilde{\mathbf{x}},\widetilde{\mathbf{Y}}=\widetilde{\mathbf{y}}\vert E_{ij}}\of{1-O\of{\frac{L_X+L_Y}{F}}}\nonumber\\
&=\sum_{\widetilde{\mathbf{x}},\widetilde{\mathbf{y}}} P\of{\widetilde{\mathbf{X}}=\widetilde{\mathbf{x}},\widetilde{\mathbf{Y}}=\widetilde{\mathbf{y}}\vert E_{ij}}\of{1-O\of{\frac{1}{\epsilon^dN}}}\nonumber\\
&=\of{1-O\of{\frac{1}{\epsilon^dN}}}\sum_{\widetilde{\mathbf{x}},\widetilde{\mathbf{y}}} P\of{\widetilde{\mathbf{X}}=\widetilde{\mathbf{x}},\widetilde{\mathbf{Y}}=\widetilde{\mathbf{y}}\vert E_{ij}}\nonumber\\
&=\of{1-O\of{\frac{1}{\epsilon^dN}}},
\end{align}
where in the fourth line we have used \eqref{bound_LX_LY}.
\end{proof}

\begin{proof}[\textbf{Proof of \ref{collision_lemma}}]  
$N'_i$ and $M'_j$ respectively are defined as the number of the input points $\mathbf{X}$ and $\mathbf{Y}$ mapped to the buckets $\widetilde{X}_i$ and $\widetilde{Y}_j$ using $H_1$. Define $\mathcal{A}_i:= \{j: H_2(\widetilde{X}_j)=i\}$ and $\mathcal{B}_i:= \{j: H_2(\widetilde{Y}_j)=i\}$. For each $i$ we can rewrite $N_i$ and $M_i$ as
\begin{align}\label{NiMi}
N_i=\sum_{j=1}^{L_X} \mathbbm{1}_{\mathcal{A}_i}(j)N_j', \hspace{2mm} M_i=\sum_{j=1}^{L_Y} \mathbbm{1}_{\mathcal{B}_i}(j)M_j'.
\end{align}

Thus, 

\begin{align}
\mathbb{B}_H &\leq \sum_{i,j\in \mathcal{F}} P\of{E_{ij}^{> 1}} \E{\mathbbm{1}_{E_{ij}}\omega_{i}\omega'_{j}\widetilde{g}\of{\omega_{ij}}\middle\vert E_{ij}^{> 1}}\nonumber\\
&= \sum_{i,j\in \mathcal{F}} P\of{E_{ij}^{> 1}}\of{P\of{E_{ij}\vert E_{ij}^{> 1}} \E{\omega_{i}\omega'_{j}\widetilde{g}\of{\omega_{ij}}\middle\vert E_{ij}^{> 1},E_{ij}}+
P\of{\overline{E_{ij}}\vert E_{ij}^{> 1}} \E{\omega_{i}\omega'_{j}\widetilde{g}\of{\omega_{ij}}\middle\vert E_{ij}^{> 1},\overline{E_{ij}}}} \nonumber\\
&= \sum_{i,j\in \mathcal{F}} P\of{E_{ij}}P\of{E_{ij}^{> 1}\vert E_{ij}} \E{\omega_{i}\omega'_{j}\widetilde{g}\of{\omega_{ij}}\middle\vert E_{ij}^{> 1},E_{ij}} \label{B_H_3rd}\\
&\leq O\of{\frac{U }{\epsilon^dN}} \sum_{i,j\in \mathcal{F}} P\of{E_{ij}} \E{\omega_{i}\omega'_{j}\middle\vert E_{ij}^{>1},E_{ij}}\label{B_H_4th}\\
&= O\of{\frac{U }{\epsilon^dN^3}} \sum_{i,j\in \mathcal{F}} P\of{E_{ij}} \E{N_{i}M_{j}\middle\vert E_{ij}^{>1},E_{ij}}\nonumber\\
&= O\of{\frac{U }{\epsilon^dN^3}} \sum_{\widetilde{\mathbf{x}},\widetilde{\mathbf{y}}} p_{\widetilde{\mathbf{X}},\widetilde{\mathbf{Y}}}\of{\widetilde{\mathbf{x}},\widetilde{\mathbf{y}}} \sum_{i,j\in \mathcal{F}} P\of{E_{ij}} \E{N_{i}M_{j}\middle\vert E_{ij}^{>1},E_{ij}, \widetilde{\mathbf{X}}=\widetilde{\mathbf{x}},\widetilde{\mathbf{Y}}=\widetilde{\mathbf{y}}}\nonumber\\
&= O\of{\frac{U }{\epsilon^dN^3}} \sum_{\widetilde{\mathbf{x}},\widetilde{\mathbf{y}}}p_{\widetilde{\mathbf{X}},\widetilde{\mathbf{Y}}} \sum_{i,j\in \mathcal{F}} P\of{E_{ij}} \E{\of{\sum_{r=1}^{L_X}\mathbbm{1}_{\mathcal{A}_i}(r)N_r'} \of{\sum_{s=1}^{L_Y}\mathbbm{1}_{\mathcal{B}_j}(s)M_s'}\middle\vert E_{ij}^{>1},E_{ij}, \widetilde{\mathbf{X}}=\widetilde{\mathbf{x}},\widetilde{\mathbf{Y}}=\widetilde{\mathbf{y}}}\label{B_H_7th}\\
&= O\of{\frac{U }{\epsilon^dN^3}} \sum_{\widetilde{\mathbf{x}},\widetilde{\mathbf{y}}} p_{\widetilde{\mathbf{X}},\widetilde{\mathbf{Y}}}\sum_{i,j\in \mathcal{F}} P\of{E_{ij}} \sum_{r=1}^{L_X}\sum_{s=1}^{L_Y}\E{\of{\mathbbm{1}_{\mathcal{A}_i}(r)} \of{\mathbbm{1}_{\mathcal{B}_j}(s)}\middle\vert E_{ij}^{>1},E_{ij}, \widetilde{\mathbf{X}}=\widetilde{\mathbf{x}},\widetilde{\mathbf{Y}}=\widetilde{\mathbf{y}}}\E{N_r'M_s'\vert E_{ij}, \widetilde{\mathbf{X}}=\widetilde{\mathbf{x}},\widetilde{\mathbf{Y}}=\widetilde{\mathbf{y}}}\nonumber\\
&= O\of{\frac{U }{\epsilon^dN^3}} \sum_{\widetilde{\mathbf{x}},\widetilde{\mathbf{y}}} p_{\widetilde{\mathbf{X}},\widetilde{\mathbf{Y}}}\sum_{i,j\in \mathcal{F}} P\of{E_{ij}} \sum_{r=1}^{L_X}\sum_{s=1}^{L_Y}P\of{r\in \mathcal{A}_i, s\in \mathcal{B}_j\big\vert E_{ij}^{>1},E_{ij}, \widetilde{\mathbf{X}}=\widetilde{\mathbf{x}},\widetilde{\mathbf{Y}}=\widetilde{\mathbf{y}}}\E{N_r'M_s'\vert E_{ij}, \widetilde{\mathbf{X}}=\widetilde{\mathbf{x}},\widetilde{\mathbf{Y}}=\widetilde{\mathbf{y}}}\label{B_H_1},
\end{align}
where in \eqref{B_H_3rd} we have used the Bayes rule, and the fact that $\widetilde{g}\of{\omega_{ij}}=0$ conditioned on the event $\overline{E_{ij}}$. In \eqref{B_H_4th} we have used the bound in Lemma \ref{X_ij}, and the upper bound on $\widetilde{g}(\omega_{ij})$ . Equation \eqref{B_H_7th} is due to \eqref{NiMi}.
Now we simplify $P\of{r\in \mathcal{A}_i, s\in \mathcal{B}_j\big\vert E_{ij}^{>1},E_{ij}, \widetilde{\mathbf{X}}=\widetilde{\mathbf{x}},\widetilde{\mathbf{Y}}=\widetilde{\mathbf{y}}}$ in \eqref{B_H_1} as follows. First assume that $\tilde{X}_r\neq\tilde{Y}_s$.

\begin{align}\label{B_H_2}
P\of{r\in \mathcal{A}_i, s\in \mathcal{B}_j\big\vert E_{ij}^{>1},E_{ij}, \widetilde{\mathbf{X}}=\widetilde{\mathbf{x}},\widetilde{\mathbf{Y}}=\widetilde{\mathbf{y}}}&\leq P\of{r\in \mathcal{A}_i, s\in \mathcal{B}_j\big\vert E_{v_i}^{>1}, E_{u_j}^{>1}, \widetilde{\mathbf{X}}=\widetilde{\mathbf{x}},\widetilde{\mathbf{Y}}=\widetilde{\mathbf{y}}} \nonumber\\
&= P\of{r\in \mathcal{A}_i\big\vert E_{v_i}^{>1}, \widetilde{\mathbf{X}}=\widetilde{\mathbf{x}}} P\of{s\in \mathcal{B}_j\big\vert E_{u_j}^{>1},\widetilde{\mathbf{Y}}=\widetilde{\mathbf{y}}}, 
\end{align}
where the second line is because the hash function $H_2$ is random and independent for different inputs.
$ P\of{r\in \mathcal{A}_i\big\vert E_{v_i}^{>1},\widetilde{\mathbf{X}}=\widetilde{\mathbf{x}}}$ in \eqref{B_H_2} can be written as
\begin{align}\label{B_H_3}
 P\of{r\in \mathcal{A}_i\big\vert E_{v_i}^{>1},\widetilde{\mathbf{X}}=\widetilde{\mathbf{x}}}&=\frac{P\of{r\in \mathcal{A}_i,E_{v_i}^{>1}\big\vert \widetilde{\mathbf{X}}=\widetilde{\mathbf{x}}}}{P\of{E_{v_i}^{>1}\big\vert \widetilde{\mathbf{X}}=\widetilde{\mathbf{x}}}}.
\end{align}
We first find $P\of{E_{v_i}^{>1}\big\vert \widetilde{\mathbf{X}}=\widetilde{\mathbf{x}}}$:

\begin{align}\label{B_H_4}
P\of{E_{v_i}^{>1}\big\vert \widetilde{\mathbf{X}}=
\widetilde{\mathbf{x}}}&=1-P\of{E_{v_i}^{=0}\big\vert \widetilde{\mathbf{X}}=
\widetilde{\mathbf{x}}}-P\of{E_{v_i}^{=1}\big\vert \widetilde{\mathbf{X}}=
\widetilde{\mathbf{x}}}\nonumber\\
&= 1-\of{\frac{F-1}{F}}^{L_X}-\of{\frac{L_X}{F}\of{\frac{F-1}{F}}^{L_x-1}}\nonumber\\
&=\frac{L_X^2}{2F^2}+o\of{\frac{L_X^2}{2F^2}}.
\end{align}
Next, we find $P\of{r\in \mathcal{A}_i,E_{v_i}^{>1}\big\vert \widetilde{\mathbf{X}}=\widetilde{\mathbf{x}}}$ in \eqref{B_H_3} as follows.
\begin{align}\label{B_H_5}
P\of{r\in \mathcal{A}_i,E_{v_i}^{>1}\big\vert \widetilde{\mathbf{X}}=\widetilde{\mathbf{x}}}
&=P\of{E_{v_i}^{>1}\big\vert r\in \mathcal{A}_i,\widetilde{\mathbf{X}}=\widetilde{\mathbf{x}}}
P\of{r\in \mathcal{A}_i\big\vert \widetilde{\mathbf{X}}=\widetilde{\mathbf{x}}}\nonumber\\
&=\of{1-\of{\frac{F-1}{F}}^{L_X-1}}\of{\frac{1}{F}}=O\of{\frac{L_X}{F^2}}
\end{align}
Thus, using \eqref{B_H_4} and \eqref{B_H_5} yields 

\begin{align}\label{B_H_6}
 P\of{r\in \mathcal{A}_i\big\vert E_{v_i}^{>1},\widetilde{\mathbf{X}}=\widetilde{\mathbf{x}}}
=O\of{\frac{1}{L_X}}.
\end{align}
Similarly, we have  
\begin{align}\label{B_H_7}
 P\of{s\in \mathcal{B}_j\big\vert E_{u_j}^{>1},\widetilde{\mathbf{Y}}=\widetilde{\mathbf{y}}}
=O\of{\frac{1}{L_Y}}.
\end{align}

Now assume the case $\tilde{X}_r=\tilde{Y}_s$. Then since $H_2(\tilde{X}_r)=H_2(\tilde{Y}_s)$, we can simplify $P\of{r\in \mathcal{A}_i, s\in \mathcal{B}_j\big\vert E_{ij}^{>1},E_{ij}, \widetilde{\mathbf{X}}=\widetilde{\mathbf{x}},\widetilde{\mathbf{Y}}=\widetilde{\mathbf{y}}}$ in \eqref{B_H_1} as

\begin{align}
P\of{r\in \mathcal{A}_i, s\in \mathcal{B}_j\big\vert E_{ij}^{>1},E_{ij}, \widetilde{\mathbf{X}}=\widetilde{\mathbf{x}},\widetilde{\mathbf{Y}}=\widetilde{\mathbf{y}}}=\delta_{ij} P\of{r\in \mathcal{A}_i\big\vert E_{v_i}^{>1}, \widetilde{\mathbf{X}}=\widetilde{\mathbf{x}},\widetilde{\mathbf{Y}}=\widetilde{\mathbf{y}}}.
\end{align}
Recalling the definition $\widetilde{\mathbf{z}}:=\widetilde{\mathbf{x}}\cup \widetilde{\mathbf{y}}$ and $L_Z:=|\widetilde{\mathbf{z}}|$, similar to  

\begin{align}\label{B_H_6:2}
P\of{r\in \mathcal{A}_i\big\vert E_{v_i}^{>1}, \widetilde{\mathbf{X}}=\widetilde{\mathbf{x}},\widetilde{\mathbf{Y}}=\widetilde{\mathbf{y}}}
=O\of{\frac{1}{L_Z}}.
\end{align}

By using equations \eqref{B_H_2}, \eqref{B_H_6}, \eqref{B_H_7} and \eqref{B_H_6:2} in \eqref{B_H_1}, we can write the following upper bound for the bias estimator due to collision. 

\begin{align}\label{Estimator_Bias_Collision}
\mathbb{B}_H &\leq O\of{\frac{U }{\epsilon^dN^3}} \sum_{\widetilde{\mathbf{x}},\widetilde{\mathbf{y}}} p_{\widetilde{\mathbf{X}},\widetilde{\mathbf{Y}}}\sum_{i,j\in \mathcal{F}} P\of{E_{ij}}  \sum_{r=1}^{L_X}\sum_{s=1}^{L_Y}\E{N_r'M_s'\vert E_{ij}, \widetilde{\mathbf{X}}=\widetilde{\mathbf{x}},\widetilde{\mathbf{Y}}=\widetilde{\mathbf{y}}}\of{O\of{\frac{1}{L_XL_Y}}+\delta_{ij}O\of{\frac{1}{L_Z}}}\nonumber\\
&= O\of{\frac{U }{\epsilon^dN^3}} \sum_{\widetilde{\mathbf{x}},\widetilde{\mathbf{y}}} p_{\widetilde{\mathbf{X}},\widetilde{\mathbf{Y}}}\sum_{i,j\in \mathcal{F}} P\of{E_{ij}}  \E{\sum_{r=1}^{L_X}N_r'\sum_{s=1}^{L_Y}M_s'\vert E_{ij}, \widetilde{\mathbf{X}}=\widetilde{\mathbf{x}},\widetilde{\mathbf{Y}}=\widetilde{\mathbf{y}}}\of{O\of{\frac{1}{L_XL_Y}}+\delta_{ij}O\of{\frac{1}{L_Z}}}\nonumber\\
&= O\of{\frac{U }{\epsilon^dN^3}} \sum_{\widetilde{\mathbf{x}},\widetilde{\mathbf{y}}} p_{\widetilde{\mathbf{X}},\widetilde{\mathbf{Y}}}\sum_{i,j\in \mathcal{F}} P\of{E_{ij}} N^2\of{O\of{\frac{1}{L_XL_Y}}+\delta_{ij}O\of{\frac{1}{L_Z}}}\nonumber\\
&= O\of{\frac{U }{\epsilon^dN^3}} \sum_{\widetilde{\mathbf{x}},\widetilde{\mathbf{y}}} p_{\widetilde{\mathbf{X}},\widetilde{\mathbf{Y}}}\of{O\of{\frac{N^2}{L_XL_Y}}+O\of{\frac{N}{L_Z}}}\sum_{i,j\in \mathcal{F}} P\of{E_{ij}} \nonumber\\
&= O\of{\frac{U }{\epsilon^dN^3}} \sum_{\widetilde{\mathbf{x}},\widetilde{\mathbf{y}}} p_{\widetilde{\mathbf{X}},\widetilde{\mathbf{Y}}}\of{O\of{\frac{N^2}{L_XL_Y}}+O\of{\frac{N}{L_Z}}}\E{\sum_{i,j\in \mathcal{F}} \mathbbm{1}_{E_{ij}}} \nonumber\\
&\leq O\of{\frac{U }{\epsilon^dN^3}} \sum_{\widetilde{\mathbf{x}},\widetilde{\mathbf{y}}} p_{\widetilde{\mathbf{X}},\widetilde{\mathbf{Y}}}\of{O\of{\frac{N^2}{L_XL_Y}}+O\of{\frac{N}{L_Z}}}\of{L_XL_Y}\nonumber\\
&\leq O\of{\frac{1}{\epsilon^dN}}.
\end{align}
\end{proof}

\subsection{Bias without Collision}

A key idea in proving Theorem \ref{bias_theorem} is to show that the expectation of the edge weights $\omega_{ij}$ are proportional to the Radon-Nikodym derivative $dP_{XY}/dP_XP_Y$ at the points that correspond to the vertices $v_i$ and $u_j$. This fact is stated in the following lemma:

\begin{lemma}\label{E_wij_lemma_main}
Under the assumptions \textbf{A1-A4}, and assuming that the density functions in \textbf{A3} have bounded derivatives up to order $q\geq 0$ we have:
\begin{align}
\E{\omega_{ij}}=\frac{dP_{XY}}{dP_XP_Y}+\mathbb{B}(N,\epsilon,q,\gamma),
\end{align}
where
\begin{align}
\mathbb{B}(N,\epsilon,q,\gamma):= 
     \begin{cases}
       O\of{\epsilon^\gamma}+O\of{\frac{1}{N\epsilon^d}}, &\quad q=0\vspace{2mm}\\
       \sum_{i=1}^q C_i\epsilon^{i}+O\of{\epsilon^q}+O\of{\frac{1}{N\epsilon^d}}, &\quad q\geq 1,
     \end{cases}
\end{align}
and $C_i$ are real constants.
\end{lemma}

Note that since $\omega_{ij}=N_{ij}N/N_iM_j$, and $N_{ij}$, $N_i$ and $N_j$ are not independent variables, deriving the expectation is not trivial.
 In the following we give a lemma that provides conditions under which the expectation of a function of random variables is close to the function of expectations of the random variables. We will use the following lemma to simplify $\E{\omega_{ij}}$.
\begin{lemma}\label{E_pass_lemma}
Assume that $g(Z_1,Z_2,...,Z_k): \mathcal{Z}_1\times...\times \mathcal{Z}_k \to R$ is a Lipschitz continuous function with constant $H_g> 0$ with respect to each of variables $Z_i, 1\leq i\leq k$. Let   $\mathbb{V}\of[Z_i]$ and $\mathbb{V}\of[Z_i|X]$ respectively denote the variance and the conditional variance of each variable $Z_i$ for a given variable $X$. Then we have

\begin{align}
&\textbf{a) }\abs{\E{g\of{Z_1,Z_2,...,Z_k}}-g\of{\E{Z_1},\E{Z_2},...,\E{Z_k}}}\leq H_g\sum_{i=1}^k\sqrt{\mathbb{V}\of[Z_i]},\\
&\textbf{b) }\abs{\E{g\of{Z_1,Z_2,...,Z_k}|X}-g\of{\E{Z_1|X},\E{Z_2|X},...,\E{Z_k|X}}}\leq H_g\sum_{i=1}^k\sqrt{\mathbb{V}\of[Z_i|X]}.
\end{align}
\end{lemma}

\begin{proof}

\begin{align}
\abs{\E{g\of{Z_1,Z_2,...,Z_k}}-g\of{\E{Z_1},\E{Z_2},...,\E{Z_k}}} &= \abs{\E{g\of{Z_1,Z_2,...,Z_k}-g\of{\E{Z_1},\E{Z_2},...,\E{Z_k}}}}\nonumber\\
&\leq\E{\abs{g\of{Z_1,Z_2,...,Z_k}-
g\of{\E{Z_1},\E{Z_2},...,\E{Z_k}}}}\label{E_pass_tri_1}\\
&\leq\mathbb{E}[\vert g\of{Z_1,Z_2,...,Z_k}-g\of{\E{Z_1},Z_2,...,Z_k}+\nonumber\\ 
&\quad+g\of{\E{Z_1},Z_2,...,Z_k}-g\of{\E{Z_1},\E{Z_2},...,Z_k}\nonumber\\
&\quad+...\nonumber\\
&\quad+g\of{\E{Z_1},\E{Z_2},...,\E{Z_{k-1}},Z_k}-g\of{\E{Z_1},\E{Z_2},...,\E{Z_k}}|]\nonumber\\
&\leq\E{\bigg\vert g\of{Z_1,Z_2,...,Z_k}-g\of{\E{Z_1},Z_2,...,Z_k}\bigg\vert}\nonumber\\ 
&\quad+\E{\bigg\vert g\of{\E{Z_1},Z_2,...,Z_k}-g\of{\E{Z_1},\E{Z_2},...,Z_k}\bigg\vert}\nonumber\\
&\quad+...\nonumber\\
&\quad+\E{\bigg\vert g\of{\E{Z_1},...,\E{Z_{k-1}},Z_k}-g\of{\E{Z_1},...,\E{Z_k}}\bigg\vert}\label{E_pass_tri_2}\\
&\leq H_g\E{\abs{Z_1-\E{Z_1}} }+H_g\E{\abs{Z_2-\E{Z_2}} }+...+H_g\E{\abs{Z_k-\E{Z_k}} }\label{E_pass_Lip}\\
&\leq H_g\sum_{i=1}^k\sqrt{\mathbb{V}\of[Z_i]}\label{E_pass_Cauchy}.
\end{align}
In \eqref{E_pass_tri_1} and \eqref{E_pass_tri_2} we have used triangle inequalities. In \eqref{E_pass_Lip} we have applied Lipschitz condition, and finally in \eqref{E_pass_Cauchy} we have used Cauchy\hyp Schwarz inequality. Since the proofs of parts (a) and (b) are similar, we omit the proof of part (b).
\end{proof}
\begin{lemma}
Define $\nu_{ij}=N_{ij}/N$, and recall the definitions $\omega_{ij}=N_{ij}N/N_iN_j$, $\omega_{i}=N_{i}/N$, and $\omega_{j}'=N_{j}/N$. Then we can write
\begin{align}
\E{\omega_{ij}}=\frac{\E{\nu_{ij}}}{\E{\omega_{i}}\E{\omega_{j}'}}+O\of{\sqrt{\frac{1}{N}}}
\end{align}
\end{lemma}
\begin{proof}
The proof follows by Lemma \ref{E_pass_lemma} and the fact that $\mathbb{V}\of[\omega_{ij}]\leq O\of{1/N}$ (proved in Lemma \ref{variance_1}).
\end{proof}

Let $x_D$ and $x_C$ respectively denote the discrete and continuous components of the vector $x$, with dimensions $d_D$ and $d_C$. Also let $f_{X_C}(x_C)$ and $p_{X_D}(x_D)$ respectively denote density and pmf functions of these components associated with the  probability measure $P_X$. Let $S(x,r)$ be the set of all points that are within the distance $r/2$ of $x$ in each dimension $i$, i.e. 
\begin{align}
S(x,r):\{x|\forall i\leq d,|X_i-x_i|<r/2\}.
\end{align}

Denote $P_r(x):=P(x\in S(x,r))$. Then we have the following lemma.

\begin{lemma}\label{P_r_lemma}
Let $r<s_{\mathcal{X}}$, where $s_{\mathcal{X}}$ is the smallest possible distance in the discrete components of the support set, ${\mathcal{X}}$.  Under the assumption \textbf{A3}, and assuming that the density functions in \textbf{A3} have bounded derivatives up to the order $q\geq 0$, we have
\begin{align}
P_r(x)=P(X_D=x_D)r^{d_C}\of{f(x_C|x_D)+\mu(r,\gamma,q,\mathbf{C}_X)},
\end{align}
where 
\begin{align}\label{mu}
\mu(r,\gamma,q,\mathbf{C}_X):= 
     \begin{cases}
       O\of{r^\gamma}, &\quad q=0\vspace{2mm}\\
       \sum_{i=1}^q C_i r^{i}+O\of{r^q}, &\quad q\geq 1.
     \end{cases}
\end{align}

In the above equation, $\mathbf{C}_X:=(C_1,C_2,...,C_q)$, and $C_i$ are real constants depending on the probability measure $P_X$.

\end{lemma}

\begin{proof}
The proof is straightforward by using \eqref{Holder_eq} for the case $q=0$ (similar to (27)-(29) in \cite{noshad_AISTAT}), and using the Taylor expansion of $f(x_C|x_D)$ for the case $q\geq 1$ (similar to (36)-(37) in \cite{noshad_AISTAT}). 
\end{proof}

\begin{lemma}\label{E_wij_lemma_app}
Let $H(x)=i, H(y)=j$. Under the assumptions \textbf{A1-A3}, and assuming that the density functions in \textbf{A3} have bounded derivatives up to the order $q\geq 0$, we have

\begin{align}
\E{\omega_{ij}|E_{ij}^{\leq 1}}=\frac{dP_{XY}}{dP_XP_Y}(x,y)+\mu(\epsilon,\gamma,q,\mathbf{C}_{XY}')+O\of{\frac{1}{\sqrt{N}}},
\end{align}
where $\mu(\epsilon,\gamma,q,\mathbf{C}'_{XY})$ is defined in \eqref{mu}.
\end{lemma}

\begin{proof}
Define $\nu_{ij}=N_{ij}/N$, and recall the definitions $\omega_{ij}=N_{ij}N/N_iN_j$, $\omega_{i}=N_{i}/N$, and $\omega_{j}'=N_{j}/N$. Using Lemma \ref{E_pass_lemma} we have

\begin{align}\label{E_wij_decompose}
\E{\omega_{ij}|E_{ij}^{\leq 1}}=\frac{\E{\nu_{ij}|E_{ij}^{\leq 1}}}{\E{\omega_{i}|E_{ij}^{\leq 1}}\E{\omega_{j}'|E_{ij}^{\leq 1}}}+O\of{\frac{1}{\sqrt{N}}}
\end{align}

Assume that $H(x)=i$. Let $\mathcal{X}$ have $d_C$ and $d_D$ continuous and discrete components, respectively. Also let $\mathcal{Y}$ have $d'_C$ and $d'_D$ continuous and discrete components, respectively. Then we can write

\begin{align}\label{E_wi}
\E{\omega_i|E_{ij}^{\leq 1}}&=\frac{1}{N}\E{N_i|E_{ij}^{\leq 1}}\nonumber\\
&=P(X\in S(x,\epsilon))\nonumber\\
&=P(X_D=x_D)\epsilon^{d_C}\of{f(x_C|x_D)+\mu(\epsilon,\gamma,q,\mathbf{C}_X)},
\end{align}
where in the third line we have used Lemma \ref{P_r_lemma}. Similarly we can write 
\begin{align}\label{E_wj}
\E{\omega'_j|E_{ij}^{\leq 1}}&=P(Y_D=y_D)\epsilon^{d'_C}\of{f(y_C|y_D)+\mu(\epsilon,\gamma,q,\mathbf{C}_X)}, \nonumber\\
\E{\nu_{ij}|E_{ij}^{\leq 1}}&=P(X_D=x_D,Y_D=y_D)\epsilon^{(d_C+d'_C)}\of{f(x_C,y_C|x_D,y_D)+\mu(\epsilon,\gamma,q,\mathbf{C}_{XY})}.
\end{align}

Using \eqref{E_wi} and \eqref{E_wj} in \eqref{E_wij_decompose} results in

\begin{align}\label{E_wij_2}
\E{\omega_{ij}|E_{ij}^{\leq 1}}=\frac{P(X_D=x_D)P(Y_D=y_D)f(x_C|x_D)f(y_C|y_D)}{P(X_D=x_D,Y_D=y_D)f(x_C,y_C|x_D,y_D)}+\mu(\epsilon,\gamma,q,\mathbf{C'}_{XY})+O\of{\frac{1}{\sqrt{N}}},
\end{align}
where $\mathbf{C'}_{XY}$ depends only on $P_{XY}$. 
Now note that using Lemma \ref{P_r_lemma}, $\frac{dP_{XY}}{dP_XP_Y}(x,y)$ can be simplified as 

\begin{align}\label{dP}
\frac{dP_{XY}}{dP_XP_Y}(x,y)=\frac{\frac{dP_{XY,r}}{dr}(x,y)}{\frac{dP_{X,r}P_{Y,r}}{dr}(x,y)}=\frac{P(X_D=x_D)P(Y_D=y_D)f(x_C|x_D)f(y_C|y_D)}{P(X_D=x_D,Y_D=y_D)f(x_C,y_C|x_D,y_D)}+\mu(\epsilon,\gamma,q,\mathbf{C''}_{XY}).
\end{align}
Finally, using \eqref{dP} in \eqref{E_wij_2} gives

\begin{align}
\E{\omega_{ij}|E_{ij}^{\leq 1}}=\frac{dP_{XY}}{dP_XP_Y}(x,y)+\mu(\epsilon,\gamma,q,\mathbf{\widetilde{C}}_{XY})+O\of{\frac{1}{\sqrt{N}}},
\end{align}
where $H(x)=i,H(y)=j$.
\end{proof}
\begin{proof}[\textbf{Proof of Lemma \ref{E_wij_lemma_main}}]
Lemma \ref{E_wij_lemma_main} is a simple consequence of Lemma \ref{E_wij_lemma_app}. We have

\begin{align}\label{E_wij_bias}
\E{\omega_{ij}} &= P\of{E_{ij}^{\leq 1}}\E{\omega_{ij}|E_{ij}^{\leq 1}}+P\of{E_{ij}^{> 1}}\E{\omega_{ij}|E_{ij}^{> 1}}.
\end{align}

Recall the definitions $\widetilde{\mathbf{X}}:=\of{\widetilde{X}_1,\widetilde{X}_2,...,\widetilde{X}_{L_X}}$  and $\widetilde{\mathbf{Y}}:=\of{\widetilde{Y}_1,\widetilde{Y}_2,...,\widetilde{Y}_{L_Y}}$ as the mapped $\mathbf{X}$ and $\mathbf{Y}$ points through $H_1$. Let $\widetilde{\mathbf{Z}}:=\widetilde{\mathbf{X}}\cup \widetilde{\mathbf{Y}}$ and $L_Z:=|\widetilde{\mathbf{Z}}|$. We first find $P\of{E_{ij}^{\leq 1}}$ as follows. For a fixed set $\widetilde{\mathbf{Z}}$ we have

\begin{align}\label{Pij_1}
P\of{E_{ij}^{\leq 1}}&=P\of{E_{v_i}^{=0}\cap E_{u_j}^{=0}}+P\of{E_{v_i}^{=0}\cap E_{u_j}^{=1}}+P\of{E_{v_i}^{=1}\cap E_{u_j}^{=0}}+P\of{E_{v_i}^{=1}\cap E_{u_j}^{=1}}\nonumber\\
&=\frac{\of(F-2)^{L_Z}}{F^{L_Z}}+\frac{L_Y\of(F-2)^{L_Z-1}}{F^{L_Z}}+\frac{L_X\of(F-2)^{L_Z-1}}{F^{L_Z}}+\frac{L_YL_X\of(F-2)^{L_Z-2}}{F^{L_Z}}\nonumber\\
&=1-O\of{\frac{L_Z}{F}}\nonumber\\
&\leq 1-O\of{\frac{L_X+L_Y}{F}}\nonumber\\
&=1-O\of{\frac{1}{\epsilon^dN}}.
\end{align}
Now note that the second term in \eqref{E_wij_bias} is the bias due to collision of $H_2$, and similar to \eqref{Estimator_Bias_Collision} it is upper bounded by $O\of{\frac{1}{\epsilon^dN}}$. Thus, \eqref{Pij_1} and \eqref{E_wij_bias} give rise to 
\begin{align}
\E{\omega_{ij}} &= \frac{dP_{XY}}{dP_XP_Y}(x,y)+\mu(\epsilon,\gamma,q,\mathbf{\widetilde{C}}_{XY})+O\of{\frac{1}{\sqrt{N}}}+ O\of{\frac{1}{\epsilon^dN}}.
\end{align}
which completes the proof.

\end{proof}


In the following lemma we make a relation between the bias of an estimator and the bias of a function of that estimator.
\begin{lemma}\label{f2Z_bias}
Assume that $g(x): \mathcal{X}\to \mathbb{R}$ is infinitely differentiable. If $\widehat{Z}$ is a random variable estimating a constant $Z$ with the bias $\mathbb{B}[\widehat{Z}]$ and the variance $\mathbb{V}[\widehat{Z}]$, then the bias of $g(\widehat{Z})$ can be written as
\begin{align}
\E{g(\widehat{Z})-g(Z)} = \sum_{i=1}^{\infty}\xi_i\of{\mathbb{B}\of[\widehat{Z}]}^i+O\of{\sqrt{\mathbb{V}\of[\widehat{Z}]}},
\end{align}
where $\xi_i$ are real constants.
\end{lemma}

\begin{proof}

\begin{align}
\E{g\of{\widehat{Z}}-g(Z)} & = g\of{\E{\widehat{Z}}}-g(Z) +\E{g\of{\widehat{Z}}-g\of{\E{\widehat{Z}}}} \nonumber\\
& = \sum_{i=1}^{\infty}\of{\E{\widehat{Z}}-Z}^i\frac{g^{(i)}(Z)}{i!} + O\of{\E{\abs{g\of{\widehat{Z}}-g\of{\E{\widehat{Z}}}}}} \nonumber\\
&=\sum_{i=1}^{\infty}\xi_i\of{\mathbb{B}\of[\widehat{Z}]}^i+O\of{\sqrt{\mathbb{V}\of[\widehat{Z}]}}.
\end{align}

In the second line we have used Taylor expansion for the first term, and triangle inequality for the second term. In the third line we have used the definition $\xi_i:=g^{(i)}(Z)/i!$, and the Cauchy\hyp Schwarz inequality for the second term.

\end{proof}

In the following we compute the expectation of the first term in \eqref{bias_proof_1} and prove Theorem \ref{bias_theorem}.

\begin{proof}[\textbf{Proof of Theorem \ref{bias_theorem}}] 

Recall that $N'_i$ and $M'_j$ respectively are defined as the number of the input points $\mathbf{X}$ and $\mathbf{Y}$ mapped to the buckets $\widetilde{X}_i$ and $\widetilde{Y}_j$ using $H_1$. Similarly, $N'_{ij}$ is defined as the number of input pairs $\of{\mathbf{X},\mathbf{Y}}$ mapped to the bucket pair $\of{\widetilde{X}_i,\widetilde{Y}_j}$ using $H_1$. Define the notations $r(i):=H_2^{-1}(i)$ for $i\in \mathcal{F}$ and $s(x):=H_1(x)$ for $x\in \mathcal{X}\cup \mathcal{Y}$. Then from \eqref{E_wij_2} since there is no collision of mapping with $H_2$ into $v_i$ and $u_j$ we have

\begin{align}\label{core_bias}
\E{\frac{N'_{s(x)s(y)}N}{N'_{s(x)}N'_{s(y)}}}=\frac{dP_{XY}}{dP_XP_Y}(x,y)+\mu(\epsilon,\gamma,q,\mathbf{\widetilde{C}}_{XY})+O\of{\frac{1}{\sqrt{N}}},
\end{align}

By using \eqref{Pij_1} and defining $\tilde{h}(x)=\tilde{g}(x)/x$ we can simplify the first term of \eqref{bias_proof_1} as

\begin{align}
\sum_{i,j\in \mathcal{F}} P\of{E_{ij}^{\leq 1}} \E{\mathbbm{1}_{E_{ij}}\omega_{i}\omega'_{j}\widetilde{g}\of{\omega_{ij}}\middle\vert E_{ij}^{\leq 1}} 
&= \of{1-O\of{\frac{1}{\epsilon^dN}}} \sum_{i,j\in \mathcal{F}} \E{\mathbbm{1}_{E_{ij}}\omega_{i}\omega'_{j}\widetilde{g}\of{\omega_{ij}}\middle\vert E_{ij}^{\leq 1}}  \nonumber\\
&= \sum_{i,j\in \mathcal{F}} \E{\mathbbm{1}_{E_{ij}}\frac{N_{i}M_{j}}{N^2}\widetilde{g}\of{\frac{N_{ij}N}{N_iM_j}}\middle\vert E_{ij}^{\leq 1}}+ O\of{\frac{1}{\epsilon^dN}}  \nonumber\\
&= \sum_{i,j\in \mathcal{F}} \E{\mathbbm{1}_{E_{ij}}\frac{N'_{r(i)}M'_{r(j)}}{N^2}\widetilde{g}\of{\frac{N'_{r(i)r(j)}N}{N'_{r(i)}M'_{r(j)}}}}+ O\of{\frac{1}{\epsilon^dN}}  \nonumber\\
&= \sum_{i,j\in \mathcal{F}} \E{\mathbbm{1}_{E_{ij}}\frac{N'_{r(i)r(j)}}{N}\widetilde{h}\of{\frac{N'_{r(i)r(j)}N}{N'_{r(i)}M'_{r(j)}}}}+ O\of{\frac{1}{\epsilon^dN}}  \nonumber\\
&= \frac{1}{N}\sum_{i,j\in \mathcal{F}} \E{N'_{r(i)r(j)}\widetilde{h}\of{\frac{N'_{r(i)r(j)}N}{N'_{r(i)}M'_{r(j)}}}}+ O\of{\frac{1}{\epsilon^dN}}  \label{Bias_wout_coll_5}\\
&= \frac{1}{N}\E{\sum_{i,j\in \mathcal{F}} N'_{r(i)r(j)}\widetilde{h}\of{\frac{N'_{r(i)r(j)}N}{N'_{r(i)}M'_{r(j)}}}}+ O\of{\frac{1}{\epsilon^dN}}  \nonumber\\
&= \frac{1}{N}\E{\sum_{i=1}^N \widetilde{h}\of{\frac{N'_{s(X)s(Y)}N}{N'_{s(X)}M'_{s(Y)}}}}+ O\of{\frac{1}{\epsilon^dN}}  \nonumber\\
&= \frac{1}{N}\sum_{i=1}^N \E{\widetilde{h}\of{\frac{N'_{s(X)s(Y)}N}{N'_{s(X)}M'_{s(Y)}}}}+ O\of{\frac{1}{\epsilon^dN}}  \nonumber\\
&= \mathbb{E}_{(X,Y)\sim P_{XY}}\of[\mathbb{E}\of[\widetilde{h}\of{\frac{N'_{s(X)s(Y)}N}{N'_{s(X)}M'_{s(Y)}}}\Bigg\vert X=x,Y=y]]+ O\of{\frac{1}{\epsilon^dN}}  \nonumber\\
&= \mathbb{E}_{(X,Y)\sim P_{XY}}\of[\frac{dP_{XY}}{dP_XP_Y}]+\mu(\epsilon,\gamma,q,\mathbf{\overline{C}}_{XY})+O\of{\frac{1}{\sqrt{N}}}+O\of{\frac{1}{\epsilon^dN}}  \label{Bias_wout_coll}.\\
\end{align}
\eqref{Bias_wout_coll_5} is due to the fact that $N'_{r(i)r(j)}=0$ if there is no edge between $v_i$ and $u_j$. Also, \eqref{Bias_wout_coll} is due to \eqref{core_bias}.

From \eqref{Bias_wout_coll} and \eqref{bias_proof_1} we obtain

\begin{align}
\E{\widetilde{I}(X,Y)}&=\E{\sum_{e_{ij}\in E_G} \omega_{i}\omega'_{j}\widetilde{g}\of{\omega_{ij}}}= \mathbb{E}_{(X,Y)\sim P_{XY}}\of[\frac{dP_{XY}}{dP_XP_Y}]+\mu(\epsilon,\gamma,q,\mathbf{\overline{C}}_{XY})+O\of{\frac{1}{\sqrt{N}}}+O\of{\frac{1}{\epsilon^dN}}.
\end{align}
Finally using Lemma \ref{f2Z_bias} results in \eqref{bias_terms}.

\end{proof}

\section*{B. Variance Proof}

In this section we first prove bounds on the variances of the edge and vertex weights and then we provide the proof of Theorem \ref{variance}.  

\begin{lemma}\label{variance_1}
Under the assumptions \textbf{A1-A4}, the following variance bounds hold true.

\begin{align}\label{variance_weights}
\mathbb{V}\of[\omega_i]&\leq O\of{\frac{1}{N}},\qquad
\mathbb{V}\of[\omega'_j]\leq O\of{\frac{1}{N}},\qquad
\mathbb{V}\of[\omega_{ij}]\leq O\of{\frac{1}{N}},\qquad
\mathbb{V}\of[\nu_{ij}]\leq O\of{\frac{1}{N}}.
\end{align}

\end{lemma}
\begin{proof}
Here we only provide the variance proof of $\omega_i$. The variance bounds of $\omega'_j$, $\omega_{ij}$ and $\nu_{ij}$ can be proved in the same way. The proof is based on Efron-Stein inequality.
Define $Z_i:=(X_i,Y_i)$. For using the Efron\hyp Stein inequality on $\mathbf{Z}:=(Z_1,...,Z_N)$, we consider another independent copy of $\mathbf{Z}$ as  $\mathbf{Z}':=(Z'_1,...,Z'_N)$ and define $\mathbf{Z}^{(i)}:=(Z_1,...,Z_{i-1},Z'_i,Z_{i+1},...,Z_N)$. Define $\omega_i(\mathbf{Z})$ as the weight of vertex $v_i$ in the dependence graph constructed by the set $\mathbf{Z}$. By applying Efron\hyp Stein inequality \cite{Noshad2017} we have

\begin{align}
\mathbb{V}\of[\omega_i] &\leq \frac{1}{2} \sum_{i=1}^N \mathbb{E}\left[\left(\omega_i\of{\mathbf{Z}}-\omega_i\of{\mathbf{Z}^{(j)}}\right)^2\right]\nonumber\\
&= \frac{1}{2N^2} \sum_{i=1}^N \mathbb{E}\left[\left(N_i\of{\mathbf{Z}}-N_i\of{\mathbf{Z}^{(j)}}\right)^2\right]\nonumber\\
&\leq \frac{1}{2N^2} O\of{N}\nonumber\\
&\leq O\of{\frac{1}{N}}.
\end{align}

In the third line we have used the fact that the absolute value of $N_i\of{\mathbf{Z}}-N_i\of{\mathbf{Z}^{(j)}}$ is at most 1.

\end{proof}

\begin{proof}[\textbf{Proof of Theorem \ref{variance}} ]

We follow similar steps as the proof of Lemma \ref{variance_1}. Define $\widehat{I}_g(\mathbf{Z})$ as the mutual information estimation using the set $\mathbf{Z}$. By applying Efron\hyp Stein inequality we have

\begin{align}
\mathbb{V}\of[\widehat{I}(X,Y)] &\leq \frac{1}{2} \sum_{k=1}^N \mathbb{E}\left[\left(\widehat{I}(\mathbf{Z})-\widehat{I}(\mathbf{Z}^{(k)})\right)^2\right]\nonumber\\
&\leq \frac{N}{2} \mathbb{E}\left[\of{\sum_{e_{ij}\in E_G} \omega_{i}\of{\mathbf{Z}}\omega'_{j}\of{\mathbf{Z}}\widetilde{g}\of{\omega_{ij}\of{\mathbf{Z}}}-\sum_{e_{ij}\in E_G} \omega_{i}\of{\mathbf{Z}^{(k)}}\omega'_{j}\of{\mathbf{Z}^{(k)}}\widetilde{g}\of{\omega_{ij}}\of{\mathbf{Z}^{(k)}}}^2\right]  \nonumber\\
&= \frac{N}{2N^4} \mathbb{E}\left[\of{\sum_{e_{ij}\in E_G} N_{i}\of{\mathbf{Z}}M_{j}\of{\mathbf{Z}}\widetilde{g}\of{\frac{N_{ij}\of{\mathbf{Z}}N}{N_{i}\of{\mathbf{Z}}M_{j}\of{\mathbf{Z}}}}-\sum_{e_{ij}\in E_G} N_{i}\of{\mathbf{Z}^{(k)}}M_{j}\of{\mathbf{Z}^{(k)}}\widetilde{g}\of{\frac{N_{ij}\of{\mathbf{Z}^{(k)}}N}{N_{i}\of{\mathbf{Z}^{(k)}}M_{j}\of{\mathbf{Z}^{(k)}}}}}^2\right]  \label{variance_main_0}\\
&\leq \frac{1}{2N^3} \mathbb{E}\left[\of{\Sigma_{n_1}+\Sigma_{n_2}+\Sigma_{m_1}+\Sigma_{m_1}+D_{n_1m_1}+D_{n_2m_2}}^2\right].\label{variance_main}
\end{align} 

Note that in equation \eqref{variance_main}, when $(X_k,Y_k)$ is resampled, at most two of $N_i$ for $i\in \mathcal{F}$ are changed exactly by one (one decrease and the other increase). The same statement holds true for $M_j$. Let these vertices be $v_{n_1}$, $v_{n_2}$, $v_{m_1}$ and $v_{m_2}$. Also the pair collision counts $N_{ij}$ are fixed except possibly $N_{n_1m_1}$ and $N_{n_2m_2}$ that may change by one. So, in the fourth line $\Sigma_{n_1}$ and $\Sigma_{n_2}$ account for the changes in MI estimation due to the changes in $N_{n_1}$ and $N_{n_2}$, and $\Sigma_{m_1}$ and $\Sigma_{m_2}$ account for the changes in $M_{m_1}$ and $v_{m_2}$, respectively. Finally $D_{n_1m_1}$ and $D_{n_2m_2}$ account for the changes in MI estimation due to the changes in $N_{n_1m_1}$ and $N_{n_2m_2}$. For example, $\Sigma_{n_1}$ is precisely defined as follows:
\begin{align}
\Sigma_{n_1}:=\sum_{j:e_{mj}\in E_G} N_{m}M_{j}\widetilde{g}\of{\frac{N_{mj}N}{N_mN_j}}-(N_{m}+1)M_{j}\widetilde{g}\of{\frac{N_{mj}N}{(N_m+1)M_j}}
\end{align}
where we have used the notations $N_i$ and $N_i^{(k)}$ instead of $N_i(\mathbf{Z})$ and $N_i(\mathbf{Z}^{(k)})$ for simplicity. Now note that by assumption \textbf{A4} we have

\begin{align}\label{g_diff}
\abs{\widetilde{g}\of{\frac{N_{mj}N}{N_mM_j}}-\widetilde{g}\of{\frac{N_{mj}N}{(N_m+1)M_j}}} &\leq G_g \abs{\frac{N_{mj}N}{N_mM_j}-\frac{N_{mj}N}{(N_m+1)M_j}}\nonumber\\
&\leq O\of{\frac{N_{mj}N}{N^2_mM_j}}.
\end{align}

Thus, using \eqref{g_diff}, $\Sigma_{n_1}$ can be upper bounded as follows

\begin{align}
\Sigma_{n_1}&\leq \sum_{j:e_{mj}\in E_G} O\of{\frac{N_{mj}N}{N_m^2}} =  O\of{\frac{N}{N_m}} \leq O(N).
\end{align}

It can similarly be shown that $N_{n_2}$, $\Sigma_{m_1}$, $\Sigma_{m_2}$, $D_{n_1m_1}$ and $D_{n_2m_2}$ are upper bounded by $O(N)$. Thus, \eqref{variance_main} simplifies as follows

\begin{align}\
\mathbb{V}\of[\widehat{I}(X,Y)] &\leq \frac{36O(N^2)}{2N^3}= O(\frac{1}{N}). 
\end{align}

\end{proof}

\section*{C. Optimum MSE Rates of EDGE}
In this short section we prove Theorem \ref{ensemble_theorem}. 

\begin{proof}[\textbf{Proof of Theorem \ref{ensemble_theorem}}]

The proof simply follows by using the ensemble theorem in (\cite{Kevin16}, Theorem 4) with the parameters $\psi_i(t)=t^{i}$ and $\phi_{i,d}(N)=N^{-i/2d}$ for the bias result in Theorem \ref{bias_theorem}. Thus, the following weighted ensemble estimator (EDGE) can achieve the optimum parametric MSE convergence rate of $O(1/N)$ for $q\geq d$.

\begin{align}
\widehat{I}_w:=\sum_{t\in \mathcal{T}}w(t)\widehat{I}_{\epsilon(t)},
\end{align}

\end{proof}

\end{document}